\documentclass[10pt,conference]{IEEEtran}
\usepackage{cite}
\usepackage{amsmath,amssymb,amsthm}
\newtheorem{definition}{Definition}
\newtheorem{assumption}{Assumption}
\newtheorem{estimation}{Estimation}

\newtheorem{lemma}{Lemma}
\newcommand{\Skel}{\mathsf{Skel}}
\newcommand{\Min}{\mathsf{Min}}
\newcommand{\Norm}{\mathsf{Norm}}

\newcommand{\Cover}{\mathsf{Cover}}
\newcommand{\Ability}{\mathsf{Ability}}
\newcommand{\Hit}{\mathsf{Hit}}
\newcommand{\Pred}{\mathsf{Pred}}
\usepackage{listings}
\usepackage[caption=false,font=footnotesize]{subfig}
\usepackage{graphicx}
\usepackage{booktabs}
\usepackage{url}
\usepackage{tikz}
\usetikzlibrary{positioning,arrows.meta,decorations.pathreplacing}
\usepackage[ruled,vlined,linesnumbered]{algorithm2e}
\DontPrintSemicolon
\SetKwProg{Fn}{Function}{:}{}
\usepackage{xcolor}
\usepackage[most]{tcolorbox}
\newtcolorbox{rqbox}{
  colback=gray!8,
  colframe=gray!50,
  boxrule=0.5pt
}
\newtcolorbox{casebox}{
  colback=blue!6,
  colframe=blue!50!black,
  boxrule=0.5pt,
  arc=2pt,
  left=6pt,
  right=6pt,
  top=4pt,
  bottom=4pt
}

\lstdefinestyle{skel}{
  basicstyle=\ttfamily\scriptsize,
  columns=fullflexible,
  breaklines=true,
  frame=single,
  xleftmargin=0.2em,
  framexleftmargin=0.2em,
  aboveskip=0.4em,
  belowskip=0.2em
}
\newenvironment{smallequation}
  {\begingroup\small\begin{equation}}
  {\end{equation}\endgroup}
\begin{document}

\title{LLM-Assisted Automatic Security Proofs for Cryptographic Protocols: How Far Are We?}

\author{
\IEEEauthorblockN{
Tianjian Liu\IEEEauthorrefmark{1}\IEEEauthorrefmark{6},
Shicheng Feng\IEEEauthorrefmark{2},
Jin'ao Shang\IEEEauthorrefmark{1},
Xiaoting Lyu\IEEEauthorrefmark{1},
Bin Wang\IEEEauthorrefmark{3},
Zonghua Zhang\IEEEauthorrefmark{4},
Lei Xue\IEEEauthorrefmark{5}\IEEEauthorrefmark{6}*,
Wei Wang\IEEEauthorrefmark{1}*
}

\IEEEauthorblockA{
\IEEEauthorrefmark{1}Xi'an Jiaotong University, Xi'an, China\\
\{tianjian.liu, jinao\_s\}@stu.xjtu.edu.cn,
\{xiaoting.lyu, wei.wang\}@xjtu.edu.cn
}

\IEEEauthorblockA{
\IEEEauthorrefmark{2}Tianjin University, Tianjin, China\\
fengshicheng@tju.edu.cn
}

\IEEEauthorblockA{
\IEEEauthorrefmark{3}Zhejiang Key Laboratory of Artificial Intelligence of Things (AIoT) Network and Data Security\\
wangbin02@xidian.edu.cn
}

\IEEEauthorblockA{
\IEEEauthorrefmark{4}CRSC Research \& Design Institute Group Co., Ltd\\
zhangzonghua@crscd.com.cn
}

\IEEEauthorblockA{
\IEEEauthorrefmark{5}Sun Yat-sen University, Guangzhou, China\\
xuelei3@mail.sysu.edu.cn
}

\IEEEauthorblockA{
\IEEEauthorrefmark{6}Shenzhen Loop Area Institute, Shenzhen, China
}

\IEEEauthorblockA{
*Corresponding authors: Lei Xue and Wei Wang.
}
}

\maketitle

\begin{abstract}
Large language models (LLMs) have shown strong potential for assisting software and security analysis tasks, yet their effectiveness in cryptographic symbolic protocol verification remains insufficiently understood.

In this paper, we conduct the first systematic evaluation of the capability of state-of-the-art LLMs in cryptographic symbolic protocol verification. To quantify this capability, we propose \textsc{CRoST} (Coverage Rate of Solve Tree), a proof-based metric derived from the verifier's proof skeleton that measures the similarity between generated lemmas and reference lemmas. We then establish the rationale of \textsc{CRoST} through both theoretical analysis and empirical validation. The evaluation results show that state-of-the-art models achieve 38.82\% coverage on average, with 14.4\% of generated lemmas exceeding 80\% coverage, indicating that LLMs can already generate useful lemmas to a certain extent. However, they still exhibit non-trivial failure modes on complex multi-phase protocols, show diminishing returns under naive scaling, and incur substantial verification overhead. These findings clarify the practical potential and limitations of LLMs for protocol verification and motivate future work on complex real-world protocols.

\end{abstract}


\section{Introduction}
\label{sec:introduction}
% 重要性
% 创新性
% 时效性
% 表述 没辙
% 工作量

Cryptographic protocols are designed to coordinate security-critical interactions and achieve security goals such as secrecy and authentication. In symbolic models, these goals are formalized as security properties over protocol executions, enabling verification tools to reason about whether an active adversary can break them. This style of analysis has proved practically valuable by exposing subtle flaws in widely deployed or standardized protocols, e.g., PCS violations enabled by Signal’s session handling layer \cite{10.5555/3620237.3620307}, multiple confirmed vulnerabilities in Apple/Samsung proximity tracking protocols \cite{10.5555/3766078.3766353}, and critical attacks on the EMV payment standard \cite{9519404}.

However, as shown in Figure \ref{fig:formal-tool-struction}, completing a protocol verification typically involves \emph{(Step 1)} formalizing the cryptographic protocol, \emph{(Step 2)} defining security properties, and \emph{(Step 3)} analyzing counterexample traces \cite{10.1016/j.ic.2007.07.002}.
This process still requires substantial expert knowledge and manual effort~\cite{10.1145/3133956.3134063, 10.5555/3766078.3766353, 10.1145/3133956.3134043, 10.5555/3620237.3620607}.
As reported by the Galois team~\cite{9733177}, formally verifying the AES-256-GCM and SHA-384 implementations in AWS LibCrypto required approximately nine person-months of work by experienced proof engineers. 

\begin{figure}[htb]
    \centering
    \includegraphics[width=0.95\linewidth]{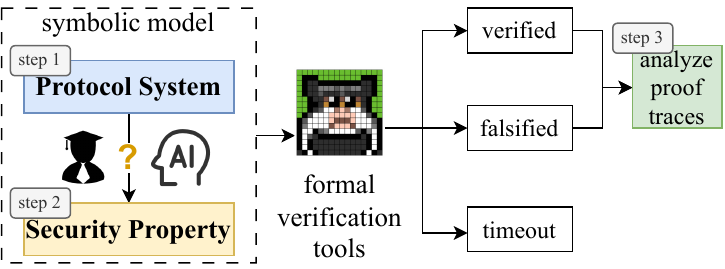}
    \caption{Symbolic model verification pipeline.}
    \label{fig:formal-tool-struction}
\end{figure}

Recent advances in LLMs \cite{pan-etal-2023-logic,NEURIPS2023_44414694,10.5555/3295222.3295349,DBLP:conf/iclr/HanRWAP22,DBLP:conf/iclr/JiangWZL0LJLW23} suggest a possible route for automatically extracting unambiguous and analyzable specifications, thereby lowering the barrier to applying formal methods. Prior work on combining LLMs with symbolic model verification tools for cryptographic protocols has mainly studied the translation of natural-language protocol descriptions into tool-checkable symbolic models \emph{(Step~1)} \cite{10.1109/ICSE55347.2025.00197}. However, the complementary task of specifying security properties \emph{(Step~2)} remains largely manual. Symbolic verifiers require explicit security goals as input, and these goals must be written as tool-checkable lemmas. Writing and refining such lemmas requires expert knowledge; it also determines what the verifier actually checks and how difficult the subsequent proof search becomes \cite{10.1007/978-3-642-39799-8_48,tamarin-manual,10.1007/978-3-030-59013-0_1,10.1007/978-3-031-17143-7_4}. We therefore focus on security-property specification for already available protocol models, and study whether LLMs can generate tool-checkable security properties in \textbf{Cryptographic Symbolic Protocol Verification (CSPV)} \emph{(Step~2)}.

To answer this question, this paper presents a systematic empirical study that focuses on two research questions:
\begin{itemize}
    \item \textbf{What are the capability boundaries of LLMs for security-property generation in CSPV?}
    \item \textbf{How do LLMs perform across different task types in CSPV?}
\end{itemize}

% %相比于fuzzing测试有明确定义的测试目标，例如触发crash，密码协议符号模型的安全属性大部分需要人为定义，并且密码学协议的安全属性定义由于场景多变，并没有像软件测试中的CWE等公开vulnerabilities数据库。传统的自动化生成安全属性方式\cite{10.1007/978-3-030-59013-0_1}局限于secrecy性质，例如识别hash或者加密函数的参数，并证明敌手不存在获得之的路径。

% 这里说明难点、该工作和其它工作不同的点，但有没有必要在这里说模糊测试和数据集不同的问题？也就是说，是不是说的有点过于详细了？挑战性能不能就改成两个点简单说明，然后在后面评测方法和实验部分再具体阐述。此外，这里少了一个少样本的特征
% 测评CSPV这一问题存在以下的难点：1.同一安全属性的形式化表述具有多样性，因此难以用简单的文本相似对比一个生成的安全属性和正确的安全属性之间的相似度。为了获得安全属性在验证过程中的深层次特征，我们引入了一种基于Tamarin证明骨架的解决级路径覆盖度量，该度量可以覆盖基于证明层的安全属性特征。我们通过突变测试从理论和实证上进一步证明了这一指标。

% 用before
There are two challenges in answering the above questions.

First, due to the diversity of representations (We will give an example in Section \ref{subsec:metric_motivatioin}) for the same security property, it is difficult to measure the similarity between a generated security property and a target correct security property using simple text-based comparison.
To capture deeper features of security properties in the verification process, we introduce a proof-based metric derived from Tamarin's proof skeleton, which reflects security properties feature at the proof-search layer \cite{tamarin-manual, 10.1007/978-3-642-39799-8_48}.
We further justify the rationale and effectiveness of this metric through theoretical analysis and mutation-based empirical validation.

% 形式化验证语法知识训练样本少，因此而产生的语法错误可能会使得LLM不能充分发挥生成安全属性的能力

Second, symbolic protocol verification tools such as Tamarin use domain-specific specification languages to describe protocols and their security properties. In such languages, security properties often involve temporal semantics, deeply nested logical structures, and protocol-specific constraints. These characteristics make it difficult for LLMs to generate syntactically correct security properties, even when the intended property is conceptually simple \cite{11334235}. In this work, we therefore investigate the effect of few-shot prompting on LLM-based security property generation and evaluate whether lightweight in-context examples can improve syntactic validity and downstream verification performance.

%这里还是需要有一个过渡句讨论RQ34

% 中文直译：基于该评估框架，我们开展实证研究，目标解决以下的几个问题：
% RQ1：LLM在CSPV的语法正确表现如何？
% RQ2：CRoST的测评方法是否有效？
% RQ3：在CSPV中，LLM的能力边界是什么？
% RQ4：在不同类型的任务下，LLM的表现如何？
% RQ5：LLM在CSPV中开销有多大？

In summary, our contributions are as follows:
% 模型的数量多、分类广是否也可以作为contribution
\begin{itemize}
  \item We present the first systematic empirical study of LLM capability in CSPV tasks of secrecy, authentication, and sanity. Our evaluation shows that LLMs exhibit stronger performance on existential propositions.
  
  \item We propose a solve-level coverage metric named CRoST (\textbf{C}overage \textbf{R}ate \textbf{o}f \textbf{S}olve \textbf{T}ree) for quantifying security properties similarity, and validate its rationality through both theoretical analysis and mutation-based evaluation.
  
  \item We unveil the main limitations and failure modes of mainstream LLMs when they are applied in CSPV tasks. We show that stage-point confusion in complex multi-phase protocols and substantial time overhead are two major bottlenecks for using LLMs in CSPV.
\end{itemize}

\section{Background}
\label{sec:background}

% 加入FOL的说明
\subsection{Symbolic Model}

Symbolic verification \cite{8115633} of cryptographic protocols is conducted under the Dolev-Yao adversary model.
In this adversary model, cryptographic primitives are treated as ideal black boxes, and the adversary is allowed to intercept, replay, modify, and inject messages arbitrarily, while being limited to symbolic inference rules derived from the algebraic properties of cryptographic operators \cite{10.1109/TIT.1983.1056650}.

\subsection{Security Properties}
\label{subsec:security_properties}

Cryptographic protocols are intended to guarantee security properties such as secrecy, authentication, and sanity. Secrecy captures scenarios in which sensitive data, such as session keys or private messages, should remain unavailable to the adversary \cite{10.1007/978-3-642-28641-4_2, tamarin-manual}. Authentication captures scenarios in which one party’s acceptance should correspond to a genuine prior action of its intended peer, rather than to an adversarial forgery or replay \cite{10.1007/978-3-031-17143-7_4, 596782, tamarin-manual}. In protocol analysis, sanity is often used as a basic executability or consistency property, checking that the modeled protocol can complete an intended interaction and that security claims do not hold only vacuously \cite{tamarin-manual}.

% Tamarin prover的两个部分rule和lemma都要给例子，能否用上introduction的图
% 要说明选择tamarin的理由
\subsection{Tamarin Prover}

In the symbolic verification tool Tamarin Prover \cite{Schmidt2012Formal, 6266153}, protocols are specified using multiset rewriting rules of the form: $[\textit{lhs}] \;\; -[\textit{ActionFacts}]\rightarrow \;\; [\textit{rhs}]$

% 中文注释：规则以“事实的多重集”作为状态：触发规则需要满足左侧事实；触发后（可消耗的）事实被移除并产生右侧事实；中间的 action facts 作为可观测事件追加到 trace，用于后续性质（lemma）定义。
Where $\textit{lhs}$ and $\textit{rhs}$ are multisets of facts representing the current protocol state.
A rule can fire when the facts in $\textit{lhs}$ are available; firing the rule consumes the facts on the left and produces the facts on the right.
The bracketed \textit{ActionFacts} on the arrow are not part of the state; instead, they are recorded as observable events on the execution trace and form the basis for specifying security properties as lemmas.
In Figure~\ref{fig:rule_example}, the rule \texttt{Reveal\_ltk} models long-term key compromise: it reuses a persistent key fact \texttt{!Ltk(A,ltk)}, emits the event \texttt{LtkReveal(A)}, and outputs the leaked key via \texttt{Out(ltk)}.

\begin{figure}[htb]
\centering
\begin{lstlisting}[
  basicstyle=\ttfamily\scriptsize,
  columns=fullflexible,
  keepspaces=true,
  breaklines=true,
  backgroundcolor=\color{gray!8},
]
rule Reveal_ltk:
  [ !Ltk(A,ltk) ]     // storing long-term key.
  --[ LtkReveal(A) ]->// Record the reveal event.
  [ Out(ltk) ]        // Output the leaked key.
\end{lstlisting}
\vspace{-0.6em}
\caption{Rewriting rule example.}
\label{fig:rule_example}
\end{figure}

% \begin{figure}[htb]
% \centering
% \begin{lstlisting}[basicstyle=\ttfamily\scriptsize,
%   columns=fullflexible,
%   keepspaces=true,
%   breaklines=true]
% rule Reveal_ltk:
%   [ !Ltk(A,ltk) ]
%   --[ LtkReveal(A) ]->
%   [ Out(ltk) ]
% \end{lstlisting}
% \vspace{-0.6em}
% \caption{Rewriting rule example.}
% \label{fig:rule_example}
% \end{figure}

Security properties in Tamarin are expressed as \emph{lemmas} (Figure \ref{fig:tamarin-secrecy-lemmas}), which are first-order logic (FOL) with quantification over timepoints.
Lemmas may quantify universally or existentially over timepoints and action facts, allowing the specification of both safety and reachability properties.
Tamarin distinguishes between \texttt{all-traces} lemmas, which must hold for all possible executions, and \texttt{exists-trace} lemmas, which assert the existence of a witness trace \cite{tamarin-manual}.

% 中文注释：我们选择 Tamarin 作为主要后端有两点原因：（1）相比常见符号协议验证工具，Tamarin 在建模与证明层面更“表达力强/功能完备”，既支持复杂等式理论（如 DH）、可变状态，也同时提供自动与交互式证明模式，便于处理真实协议中的复杂证明义务；（2）Tamarin 生态支持 SAPIC+ 过程演算前端，可将同一协议规格编译到不同验证后端（如 ProVerif、DeepSec），便于跨工具复核与复现。
In this work, we instantiate our evaluation on Tamarin as the target symbolic verifier for two reasons.

First, among symbolic protocol verification tools, Tamarin offers particularly strong modeling expressiveness and proof support. It can handle rich equational theories, mutable state, and provides both automatic and interactive proof modes, which is important for complex, real-world protocol models \cite{CremersTools}.
Second, Tamarin integrates the SAPIC+ process calculus~\cite{281300}, a cross-tool modeling layer that can translate SAPIC+ specifications to Tamarin, ProVerif, and DeepSec. This allows protocol models developed in a multi-verifier workflow to be instantiated as Tamarin theories for our evaluation.

\section{Methodology}
\label{sec:methodology}

\subsection{Motivation}
\label{subsec:metric_motivatioin}
% background

% 这部分background讲清楚了是否就可以不讲
% 思考如何和之前衔接，能不能加一句话
% 中文注释：在 Tamarin 中，协议的安全性质由 lemmas 给出；在我们的基准中，每个协议理论都包含一组由专家撰写的 reference lemmas，用以表达设计者意图中的安全属性（如 secrecy、authentication、accountability）。因此，我们以这组 reference lemmas 作为对照目标，在相同协议模型与相同证明配置下，对 LLM 生成的 lemmas 进行比较与评估。

In this section, we assess how closely LLM-generated lemmas align with expert-written lemmas in expressing the intended security properties of a protocol. Our benchmark provides, for each protocol theory, a reference lemma set written by human experts, which we treat as the target specification. Based on this comparison, we first derive a metric CRoST from empirical observations about proof behavior and generation difficulty, and then justify its rationale in subsection \ref{subsec:def-assump}.

% 中文直译：困难在于：Tamarin lemma 是“带量词的时序一阶逻辑公式”，其语义依赖于 action facts 及时间点顺序；而同一安全属性往往没有唯一、规范的 lemma 形式，常见地会被重写、拆分或合并成不同但语义等价的 lemma 集合，因此很难建立 lemma 的一一对应关系。图ref{fig:tamarin-secrecy-lemmas}给出了一个例子：同一个“密钥保密性”的定义既可以写成一个总 lemma Secrecy，也可以拆成两个按角色泄露条件分支的 lemmas Secrecy\_split\_A 和 Secrecy\_split\_B。
To derive the metric, a key difficulty is that a Tamarin lemma is a quantified temporal FOL formula over action facts and their ordering on traces \cite{tamarin-manual, 11334235}.
There is rarely a canonical lemma form for a given security property: One security property can be expressed in various forms.
For example, Figure~\ref{fig:tamarin-secrecy-lemmas} shows that the same secrecy property definition can be written either as a single overall lemma \texttt{Secrecy}, or equivalently split into two lemmas \texttt{Secrecy\_split\_A} and \texttt{Secrecy\_split\_B} that branch on role-specific reveal conditions. This poses a challenge for lemma evaluation metrics.

\begin{figure}[htb]
\centering
\begin{lstlisting}[
  basicstyle=\ttfamily\scriptsize,
  columns=fullflexible,
  keepspaces=true,
  breaklines=true,
  backgroundcolor=\color{gray!8},
]
// Same secrecy property, written in two forms.
lemma Secrecy:
  "All A B m #i.
     Secret(A,B,m)@#i
     ==> not (Ex #r. K(m)@#r)
       | (Ex #r. Reveal(A)@#r)
       | (Ex #r. Reveal(B)@#r)"

lemma Secrecy_split_A:
  "All A B m #i.
     Secret(A,B,m)@#i
     & not(Ex #r. Reveal(A)@#r)
     ==> not (Ex #r. K(m)@#r)
       | (Ex #r. Reveal(B)@#r)"

lemma Secrecy_split_B:
  "All A B m #i.
     Secret(A,B,m)@#i
     & not(Ex #r. Reveal(B)@#r)
     ==> not (Ex #r. K(m)@#r)
       | (Ex #r. Reveal(A)@#r)"
\end{lstlisting}
\vspace{-0.6em}
\caption{Alternative Tamarin specifications for the same secrecy property.}
\label{fig:tamarin-secrecy-lemmas}
\end{figure}

% 中文直译：在这种“非规范化 + 可重构”的特征下，传统的基于文本相似度的指标（如 BLEU）无法可靠衡量属性定义是否一致；相关研究也表明，匹配类指标往往难以反映真实语义/功能正确性。
Text-level similarity metrics such as BLEU \cite{papineni-etal-2002-bleu} are not reliable signals for whether two lemma sets define the same security property; metrics based on matching are known to correlate poorly with semantic or functional correctness in other generation settings \cite{Evtikhiev2022OutOT}.
A seemingly stronger alternative is to compare FOL components via truth tables \cite{yang-etal-2024-harnessing}. However, Tamarin lemmas involve quantification over events and timepoints. Truth-table comparison is not applicable in general.

We therefore adopt a proof-based comparison based on Tamarin's proof artifacts.
Given a protocol and a lemma, Tamarin's automated proof search produces a \emph{proof skeleton} (Figure~\ref{fig:skel-contrast}) that explicitly records the encountered \texttt{solve(...)} goals and case distinctions \cite{tamarin-manual} as tree structure.
Intuitively, these solve obligations reflect which parts of the verifier's reasoning space are exercised by a lemma specification.
This is similar to proof-based criteria in software testing, where coverage quantifies how thoroughly executions exercise a graph-structured behavior space \cite{10.5555/3155562.3155590, 11081727, 10.1145/3477579}.
From a software testing perspective, a useful coverage criterion should be measurable, reliable, and predictive. Motivated by this principle, we design CRoST as a proof-based coverage metric over Tamarin proof skeletons.

% 中文直译：我们据此定义CRoST：把目标 lemma 集合的 proof skeleton 中出现的 solve-goal 路径当作待覆盖的“测试需求”，看生成 lemma 集合的 skeleton 是否覆盖这些路径（以连续子路径匹配为准）；这一思想类比软件测试中的路径覆盖准则，用覆盖率刻画结构性探索程度。
Algorithm~\ref{alg:crost} computes \textsc{CRoST} in three steps. 
Lines~1-10 define the procedure ExtractPaths, which traverses a proof tree and records every path. 
Lines~12-17 apply this procedure to all target and generated proof trees to build the path sets $P^\star$ and $P$. 
Lines~19-23 then compare each target path against the generated paths and mark it as covered when a match is found. 
Finally, Lines~24-27 return the coverage ratio.

\begin{algorithm}[htb]
\caption{CRoST: Coverage Rate of Solve Tree.}
\label{alg:crost}
\LinesNumbered
\SetCommentSty{footnotesize}

\KwIn{Target proof trees $\mathcal{T}^\star$, generated proof trees $\mathcal{T}$}
\KwOut{CRoST score $c \in [0,1]$}

\SetKwFunction{Extract}{ExtractPaths}

\Fn{\Extract{$node, prefix$}}{
    $P \gets \emptyset$\;
    \uIf{$node$ is a \texttt{solve} node}{
        $prefix' \gets prefix \,\|\, \langle label(node)\rangle$\;
        $P \gets P \cup \{prefix'\}$\;
    }
    \Else{
        $prefix' \gets prefix$\;
    }
    \ForEach{$child \in Children(node)$}{
        $P \gets P \cup \Extract(child, prefix')$\;
    }
    \Return $P$\;
}

\BlankLine
$P^\star \gets \emptyset,\; P \gets \emptyset$\;

\tcp*[l]{Extract all root-to-\texttt{solve} paths from target lemmas.}
\ForEach{$tree \in \mathcal{T}^\star$}{
    \ForEach{$root \in Roots(tree)$}{
        $P^\star \gets P^\star \cup \Extract(root, \langle\ \rangle)$\;
    }
}

\tcp*[l]{Extract all root-to-\texttt{solve} paths from generated lemmas.}
\ForEach{$tree \in \mathcal{T}$}{
    \ForEach{$root \in Roots(tree)$}{
        $P \gets P \cup \Extract(root, \langle\ \rangle)$\;
    }
}

$covered \gets \emptyset$\;
\tcp*[l]{Compute covered target paths using subsequence matching ($p \preceq q$).}
\ForEach{$p \in P^\star$}{
    \ForEach{$q \in P$}{
        \If{$p \preceq q$}{
            $covered \gets covered \cup \{p\}$\;
            \textbf{break}\;
        }
    }
}

\uIf{$|P^\star| = 0$}{
    \Return $c \gets 0$\;
}
\Else{
    \Return $c \gets \frac{|covered|}{|P^\star|}$\;
}
\end{algorithm}

% 中文直译：我们的评估指标 CRoST 的灵感来自软件测试中的结构覆盖准则：在控制流图/程序图上，以“被覆盖的节点/边/路径比例”衡量测试集对行为空间的探索程度；但我们的场景并不是执行程序来触发分支，而是运行 Tamarin 的自动证明搜索来触发 solve 义务与分支，因此覆盖对象从“程序路径”变为“证明骨架中的 solve 树路径”。由于覆盖率在不同场景下与有效性之间并非天然等价（需要具体论证），我们接下来将从理论上建立 CRoST 与“安全属性定义/检验能力”的关系，并通过突变测试在实证上验证其合理性。
% \noindent
% \textbf{From software-test coverage to CRoST.}
% Our metric CRoST is inspired by structural coverage criteria in software testing, where the adequacy of a test suite is quantified by the fraction of graph-structured requirements (e.g., nodes/edges/paths in a control-flow graph) that are exercised \cite{ammann2008introduction,zhu1997testadequacy}.
% In our algorithm, however, there is no program execution and no concrete test inputs; instead, a generated lemma set serves as the stimulus that drives Tamarin's heuristic proof search, and the resulting proof skeleton exposes which solve obligations and branch structures are explored.
% Because coverage is only a proxy signal whose connection to effectiveness depends on the underlying process and must be justified \cite{inozemtseva2014coverage}, we next establish a principled link between CRoST and a generator's capability to define/check security properties, and we validate this link quantitatively via mutation-based evaluation.
% We then prove a relationship between the \textsc{CRoST} metric and security property similarity.

Prior work has shown that proof-based coverage metrics can serve as a practical alternative to mutation-based coverage in formal verification, while yielding meaningful information about specification adequacy~\cite{10.5555/3155562.3155590}. In the symbolic verification of cryptographic protocols scenario, we go one step further by studying whether coverage also correlates with the similarity between generated lemmas and target lemmas.
% 这里的承接很重要，需要说明我们为什么要理论证明其关系，但还没想好咋写
% 这里可以写cite{10.5555/3155562.3155590}中证明了软件测试的proof-based方法和mutation test之间的关系，但是在该场景下我们希望进一步证明覆盖率和形式化安全属性表述相似性的相关性

\paragraph*{Goal and proof idea}
Our goal is to evaluate an LLM-generated lemma set by how effectively it reproduces the proof obligations required to establish a target security property in Tamarin.
We use \emph{solve-level coverage} to measure how much of the target lemma's proof skeleton is covered by the generated lemma.
In the following, we formalize (i) capability as covering a minimal sufficient set of solve obligations, i.e., a minimal generator of the sufficient-set family, and (ii) how increasing coverage raises the probability of covering at least one such minimal generator.
We further quantify this relationship under a simple random-overlap model.

\subsection{Definitions and Assumptions}
\label{subsec:def-assump}

\paragraph*{Tool-grounded observables under a fixed proving configuration}
We fix a protocol theory $P$ and a proving configuration $\kappa$ (Tamarin version, command-line flags, heuristic rules, and resource bounds).
For any lemma $L$, let $\Skel(L):=\mathsf{Autoprove}(P,L;\kappa)$ denote the proof skeleton produced by Tamarin under $\kappa$.
We view $\Skel(L)$ as a finite rooted structure whose nodes include \texttt{solve($c$)} goals and whose edges reflect explicit case distinctions in the skeleton output.

\begin{definition}[Atomic solve-item and solve footprint]
Each \texttt{solve($c$)} node contains a goal/constraint term $c$.
We apply a normalization function $\Norm(\cdot)$ that (i) $\alpha$-renames bound variables, (ii) erases proof-local indices/counters, and (iii) canonicalizes commutative constructs.
A normalized \emph{atomic solve-item} is $u:=\Norm(c)$.
Let $\mathcal{U}$ be the finite universe of normalized solve-items observed under the fixed configuration $\kappa$.
The \emph{solve footprint} of lemma $L$ is
\begin{smallequation}
S(L)
:=
\{\, \Norm(c)\in\mathcal{U}
\mid
\text{\texttt{solve($c$)} occurs in }\Skel(L)
\,\}
\subseteq \mathcal{U}.
\end{smallequation}
\end{definition}

\paragraph*{Proof-skeleton order over solve occurrences}
After normalization, each \texttt{solve($c$)} node is labeled by an atomic solve-item $u=\Norm(c)$.
For $x,y \in \mathcal{U}$, define
\begin{smallequation}
x \preceq_\kappa y
\quad \text{iff} \quad
x=y \text{ or } x \text{ is a descendant of } y
\end{smallequation}
This relation captures the solve-reduction structure induced by the proof tree: a node is below another precisely when it arises from reducing subgoals generated by it.

\paragraph*{Deterministic proof search policy}
Under fixed $\kappa$, Tamarin's automated proof search uses a fixed heuristic ranking to prioritize open constraints and proof methods.
Operationally, we treat $\Skel(L)$ as the deterministic unfolding of the proof-search policy induced by $\kappa$, yielding a reproducible solve/case structure for a given $(P,L)$.

\paragraph*{(B) Lemma-set abstraction for security properties}
In practice, a target security property may be specified by a set of Tamarin lemmas rather than a single canonical lemma.
Likewise, an LLM may generate multiple lemmas that jointly describe parts of the same intended property.
Therefore, we treat a lemma set as the unit of comparison: multiple lemmas are abstracted into a property-level solve footprint, obtained by aggregating their normalized solve footprints.

\begin{definition}[solve-level coverage]
\label{def:lemma-set-coverage}
For a finite lemma set $\mathcal{L}$ under the same protocol $P$ and proving configuration $\kappa$, we lift the solve footprint from individual lemmas to the lemma-set level by union:
\begin{smallequation}
S(\mathcal{L})
:=
\bigcup_{L\in\mathcal{L}} S(L)
\subseteq \mathcal{U}.
\end{smallequation}
Given a target lemma set $\mathcal{L}^{\star}$ and a candidate generated lemma set $\mathcal{L}$, let
\begin{smallequation}
S^{\star}:=S(\mathcal{L}^{\star}),
\qquad
S:=S(\mathcal{L}),
\qquad
H:=S\cap S^{\star}.
\end{smallequation}
We write
\begin{smallequation}
n:=|S^{\star}|,
\qquad
k:=|H|.
\end{smallequation}
The solve-level coverage is defined as
\begin{smallequation}
\Cover(\mathcal{L};\mathcal{L}^{\star})
:=
\frac{|S(\mathcal{L})\cap S(\mathcal{L}^{\star})|}
{|S(\mathcal{L}^{\star})|}
=
\frac{k}{n}
=
c\in[0,1].
\end{smallequation}
\end{definition}

\begin{definition}[Minimal generators of the sufficient-set family]
Under fixed $P, \kappa$, let
$\mathcal{F}\subseteq 2^{\mathcal{U}}$ be the upward-closed family of solve-item sets sufficient for the target outcome.
The \emph{minimal generators} of $\mathcal{F}$ are the inclusion-minimal elements of $\mathcal{F}$:
\begin{smallequation}
\Min(\mathcal{F})
:=
\{\, C\in\mathcal{F}
\mid
\nexists C'\subsetneq C:\ C'\in\mathcal{F}
\,\}.
\end{smallequation}
We write $\Min(\mathcal{F})=\{C_1,\dots,C_M\}$.
Equivalently, for any $X\subseteq\mathcal{U}$,
\begin{smallequation}
X\in\mathcal{F}
\quad\Longleftrightarrow\quad
\exists j\in\{1,\dots,M\}: C_j\subseteq X.
\end{smallequation}
Each $C_j$ represents a minimal sufficient set of normalized solve-items for the target outcome, i.e., no proper subset of $C_j$ is sufficient.
Typically $|C_j|>1$ for all-traces properties, while exists-trace properties may admit $|C_j|=1$.
Figure~\ref{fig:skel-contrast} contrasts two typical skeleton shapes: exists-trace lemmas succeed by finding one witness branch, whereas all-traces lemmas require discharging all case splits.
\end{definition}

\begin{figure*}[htb]
  \centering
  \subfloat[\textbf{exists-trace: one successful branch suffices.}\label{fig:skel-exists}]{%
    \includegraphics[width=0.49\textwidth]{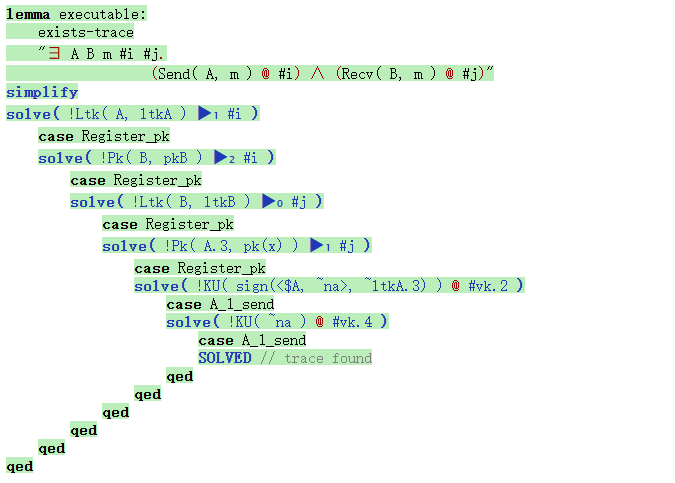}}
  \hfill
  \subfloat[\textbf{all-traces: all branches must be discharged.}\label{fig:skel-all}]{%
    \includegraphics[width=0.49\textwidth]{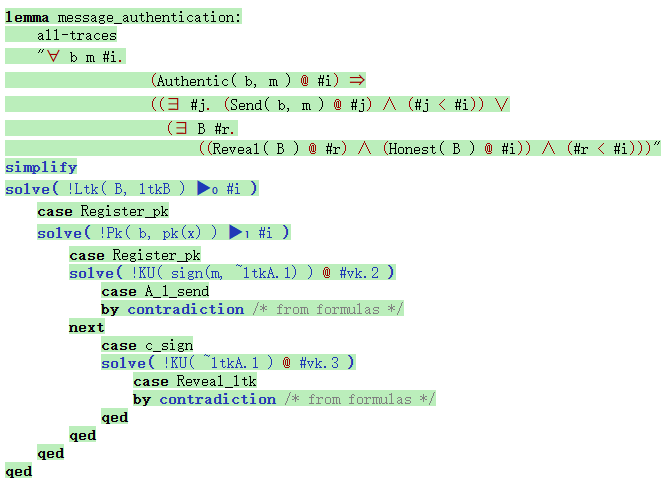}}
  \caption{\textbf{Contrasting proof skeleton for exists-trace vs all-traces lemmas in Tamarin.}}
  \label{fig:skel-contrast}
\end{figure*}

\begin{definition}[Capability of a lemma generator]
Let $\mathcal{G}$ be a lemma generator pipeline. Given a protocol context, $\mathcal{G}$ induces a distribution over finite lemma sets.
The generation capability of $\mathcal{G}$ is the probability that the generated lemma set covers at least one minimal generator:
\begin{smallequation}
\begin{split}
\Ability(\mathcal{G};\mathcal{F})
&:= \Pr_{\mathcal{L}\sim\mathcal{G}}
   \bigl[\, S(\mathcal{L})\in\mathcal{F}\,\bigr] \\
&= \Pr_{\mathcal{L}\sim\mathcal{G}}
   \bigl[\,\exists j\in\{1,\dots,M\}:
   C_j\subseteq S(\mathcal{L})\,\bigr].
\end{split}
\end{smallequation}
\end{definition}
$\mathcal{L}\sim\mathcal{G}$ means $\mathcal{L}$ is sampled from distribution induced by $\mathcal{G}$.

\paragraph*{Sanity check}
If $\Cover(\mathcal{L};\mathcal{L}^\star)=1$, then by definition
$S(\mathcal{L}^\star) \subseteq S(\mathcal{L})$.
If the target lemma set is sufficient, i.e.,
$S(\mathcal{L}^\star) \in \mathcal{F}$, then there exists some minimal generator
$C_j \subseteq S(\mathcal{L}^\star)$.
Consequently, $C_j \subseteq S(\mathcal{L})$, which implies
$S(\mathcal{L}) \in \mathcal{F}$.

\paragraph*{(C) A tractable random model}
To quantitatively relate solve-level coverage to generation capability, we adopt a simplified conditional random model.

\begin{assumption}[Random overlap conditional on coverage]
\label{assm:rand-overlap}
Fix a target lemma set $\mathcal{L}^\star$.
In practice, the generated footprint $S(\mathcal{L})$ may be correlated with $S(\mathcal{L}^\star)$ because both are induced by the same protocol theory and proving configuration.
For a tractable analysis, however, we do not assume any favorable target-specific alignment.
Conditioned only on the overlap size $|H|=k$, we model $H$ as a uniformly random $k$-subset of $S(\mathcal{L}^\star)$, sampled without replacement.
\end{assumption}

\subsection{Predecessor Hits and Minimal-Generator Continuation}
\label{subsec:predecessors}

\paragraph*{Minimal-generator elements}
For each $C_j\in\Min(\mathcal{F})$, let
\begin{smallequation}
C_j=\{u_{j,1},u_{j,2},\dots,u_{j,m_j}\}\subseteq \mathcal{U},
\qquad m_j:=|C_j|.
\end{smallequation}

\paragraph*{Predecessor sets}
Using the proof-skeleton order $\preceq_\kappa$, for each $u_{j,k}\in C_j$ let its target-derived predecessor set be
\begin{smallequation}
\Pred_{\kappa}(u_{j,k})
:=
\{\, g\in S^\star \mid u_{j,k}\preceq_\kappa g,\ g\neq u_{j,k}\,\}.
\end{smallequation}
Intuitively, $\Pred_{\kappa}(u_{j,k})$ contains target solve-items whose reduction can lead to the solve obligation $u_{j,k}$.

\begin{definition}[Predecessor-hit event]
\label{def:pred-hit}
For a generated lemma set $\mathcal{L}$, let
\begin{smallequation}
\Hit_{j,k}(\mathcal{L})=1
\quad\text{iff}\quad
H\cap\Pred_{\kappa}(u_{j,k})\neq\emptyset .
\end{smallequation}
That is, the generated lemma set hits at least one target predecessor of the minimal-generator element $u_{j,k}$.
\end{definition}

\begin{estimation}[Predecessor-to-minimal-generator continuation probability]
\label{est:rho}
For each $(j,k)$, let
\begin{smallequation}
\rho_{j,k}
:=
\Pr_{\mathcal{L}\sim\mathcal{G}}
\!\Big[\, C_j \subseteq H \ \Big|\ \Hit_{j,k}(\mathcal{L})=1 \,\Big].
\end{smallequation}
In our experiments under a fixed $\kappa$, we estimate $\rho_{j,k}$ from repeated runs of $\mathcal{G}$; for example,
the average continuation probability is $\widehat{\rho}\approx 0.513$, i.e., conditioned on a predecessor hit, the full
minimal generator is covered about $51.2\%$ of the time.
\end{estimation}

\paragraph*{Empirical predecessor-to-generator continuation probability}
A predecessor hit does not fix the remaining proof-search trajectory: the subsequent goal ordering and case splits may still diverge from the target skeleton.
Therefore, even under a fixed configuration $\kappa$, hitting a predecessor does not guarantee covering the whole minimal generator.
The central problem in this proof model can be viewed as analogous to probabilistic reachability in uncertain graphs~\cite{10.14778/3324301.3324304}: a predecessor hit indicates a possible support path, while $\rho$ estimates the probability that this path continues to full minimal-generator coverage.

% =========================
% Exists-trace special case
% =========================
\subsection{Exists-trace Situation}
\label{subsec:exists-singleton}

We first focus on the simplest \texttt{exists-trace} situation, where the target outcome can be witnessed by reaching one solve-item from a singleton minimal generator.
Let
\begin{smallequation}
W:=\{\,u\in\mathcal{U}\mid \{u\}\in\Min(\mathcal{F})\,\}
\end{smallequation}
be the set of solve-items that form singleton minimal generators.
For example, Figure~\ref{fig:exists-trace-singleton} illustrates the case where one such generator is $C_1=\{u_1\}$.

\begin{figure}[htb]
  \centering
  \includegraphics[width=0.75\linewidth]{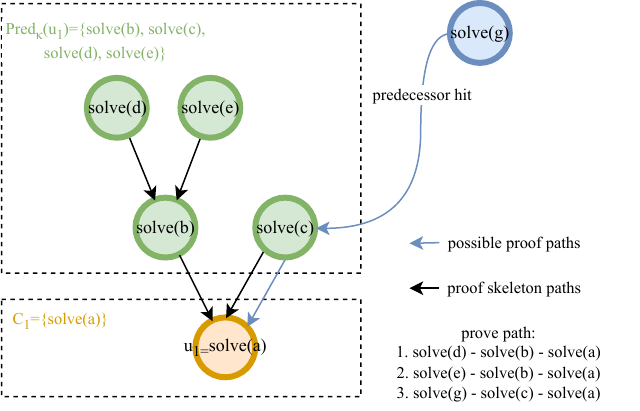}
  \caption{\textbf{Exists-trace witness shape with a singleton minimal generator.}
  The minimal generator is $C_1=\{\texttt{solve}(a)\}$; reaching this solve-item along one witness branch suffices to establish the target outcome.}
  \label{fig:exists-trace-singleton}
\end{figure}

\begin{lemma}[Capability equals hitting a singleton minimal generator]
\label{lem:exists-singleton-capability}
In the singleton \texttt{exists-trace} setting, the capability of $\mathcal{G}$ is
\begin{smallequation}
\Ability(\mathcal{G};\mathcal{F})
=
\Pr_{\mathcal{L}\sim\mathcal{G}}
\!\big[\,\exists u\in W:\ u\in S(\mathcal{L})\,\big].
\end{smallequation}
\end{lemma}

\begin{lemma}[Closed-form capability from coverage in the singleton \texttt{exists-trace} case]
\label{lem:exists-singleton-closedform}
Define the direct-hit event
\begin{smallequation}
A:=\{\,H\cap W\neq\emptyset\,\},
\end{smallequation}
i.e., the generated footprint directly covers at least one singleton minimal-generator element.

Under Assumption~\ref{assm:rand-overlap}, with $k$ treated as a fixed overlap size, we write $\Pr_k[\cdot]$
for the probability induced by drawing $H$ uniformly from all $k$-subsets of $S^\star$.
Let $r:=|W|$. Then
\begin{smallequation}
\Pr_k[A]
=
1-
\frac{\binom{n-r}{k}}{\binom{n}{k}}.
\end{smallequation}

Let $\widehat{\rho}$ denote the empirical continuation probability that, when no singleton minimal-generator element is directly hit, predecessor-supported proof search still continues to full minimal-generator coverage under fixed configuration $\kappa$.
Then the success probability is modeled as
\begin{smallequation}
\begin{aligned}
\Ability(\mathcal{G};\mathcal{F})
&=
\Pr_k[A] + \bigl(1-\Pr_k[A]\bigr)\widehat{\rho} \\
&=
\widehat{\rho} + (1-\widehat{\rho})
\left(
1-
\frac{\binom{n-r}{k}}{\binom{n}{k}}
\right).
\end{aligned}
\end{smallequation}
For the special case $r=1$, this reduces to
\begin{smallequation}
\Ability(\mathcal{G};\mathcal{F})
=
\widehat{\rho} + (1-\widehat{\rho})\,c.
\end{smallequation}
\end{lemma}

\noindent
\textbf{Interpretation.}
In the singleton \texttt{exists-trace} setting, directly covering any singleton minimal-generator element is sufficient for success.
If the generated footprint instead only reaches predecessor-supported items, $\widehat{\rho}$ models the probability that this support continues to full minimal-generator coverage.
When there is only one singleton minimal generator, capability increases linearly with coverage $c$.

% =========================
% All-traces special case
% =========================
\subsection{All-traces situation}
\label{subsec:alltraces-multi}

We now consider the typical \texttt{all-traces} situation, where the target outcome may admit
multiple minimal generators and each minimal generator may contain multiple solve-items.

\paragraph*{Mixed direct hits and predecessor hits}
For each element $u_{j,t}\in C_j$, define its support set as
\[
X_{j,t}
:=
\{u_{j,t}\}\cup \operatorname{Pred}_{\kappa}(u_{j,t})
\subseteq S^\star,
\qquad
s_{j,t}:=|X_{j,t}|.
\]
We say that $u_{j,t}$ is supported by the generated footprint if $H$ hits either the element itself
or one of its target-derived predecessors:
\[
E_{j,t}:=\{\,H\cap X_{j,t}\neq\emptyset\,\}.
\]
Accordingly, a minimal generator $C_j$ is predecessor-supported if all of its elements are supported:
\[
E_j:=\bigcap_{t=1}^{m_j}E_{j,t}.
\]
This definition allows mixed evidence: some elements of $C_j$ may be directly covered by $H$, while others may only be supported through predecessor hits.

\paragraph*{Continuation probability for an all-traces minimal generator}
Even when $E_j$ holds, the remaining proof-search trajectory may still diverge from the target skeleton.
We therefore define a minimal-generator-level continuation probability:
\begin{smallequation}
\label{eq:rho-all}
\rho_j
:=
\Pr_{\mathcal{L}\sim\mathcal{G}}
\!\Big[\, C_j\subseteq S(\mathcal{L}) \ \Big|\ E_j \,\Big]
\in[0,1].
\end{smallequation}

\begin{lemma}[Capability lower bound via predecessor-supported minimal generators]
\label{lem:alltraces-ability-lb}
For any $j\in\{1,\dots,M\}$,
\[
\Ability(\mathcal{G};\mathcal{F})
\ge
\Pr_{\mathcal{L}\sim\mathcal{G}}\!\big[\, C_j\subseteq S(\mathcal{L})\,\big]
\ge
\rho_j\cdot
\Pr_{\mathcal{L}\sim\mathcal{G}}\!\big[\,E_j\,\big].
\]
Consequently,
\[
\Ability(\mathcal{G};\mathcal{F})
\ge
\max_{j\in[M]}
\rho_j\cdot
\Pr_{\mathcal{L}\sim\mathcal{G}}\!\big[\,E_j\,\big].
\]
\end{lemma}

\begin{proof}
The first inequality holds because
\[
\Ability(\mathcal{G};\mathcal{F})
=
\Pr_{\mathcal{L}\sim\mathcal{G}}
\bigl[\exists j':\, C_{j'}\subseteq S(\mathcal{L})\bigr],
\]
which dominates any fixed witness $j$.
For the second inequality, write
\begin{smallequation}
\begin{split}
\Pr[C_j\subseteq S(\mathcal{L})]
&=
\Pr[C_j\subseteq S(\mathcal{L})\mid E_j]\Pr[E_j] \\
&\quad+
\Pr[C_j\subseteq S(\mathcal{L})\mid \neg E_j]\Pr[\neg E_j] \\
&\ge
\rho_j\Pr[E_j],
\end{split}
\end{smallequation}
using the definition of $\rho_j$ in~\eqref{eq:rho-all} and non-negativity of probabilities.
\end{proof}

\paragraph*{A computable lower bound for $\Pr[E_j]$ under random overlap}
Under Assumption~\ref{assm:rand-overlap}, $H$ is a uniformly random $k$-subset of $S^\star$.
For each element, define the miss event
\[
F_{j,t}:=\{\,H\cap X_{j,t}=\emptyset\,\}.
\]
Then $E_j=\bigcap_t\neg F_{j,t}$ and hence $\neg E_j=\bigcup_t F_{j,t}$.
By a union bound,
\begin{smallequation}
\label{eq:Ej-unionbound}
\Pr[E_j]\ge 1-\sum_{t=1}^{m_j}\Pr[F_{j,t}].
\end{smallequation}
Moreover, each miss probability has a hypergeometric closed form:
\begin{smallequation}
\label{eq:Fjt-exact}
\Pr[F_{j,t}]
=
\frac{\binom{n-s_{j,t}}{k}}{\binom{n}{k}},
\end{smallequation}
with the convention that $\binom{a}{k}=0$ if $a<k$.
Plugging~\eqref{eq:Fjt-exact} into~\eqref{eq:Ej-unionbound} yields
\begin{smallequation}
\label{eq:Ej-lb-final}
\Pr[E_j]
\ge
1-\sum_{t=1}^{m_j}
\frac{\binom{n-s_{j,t}}{k}}{\binom{n}{k}}.
\end{smallequation}

\paragraph*{Coverage-to-capability relation}
Combining Lemma~\ref{lem:alltraces-ability-lb} and~\eqref{eq:Ej-lb-final}, and substituting $k=cn$,
we obtain the following monotone lower bound:
\begin{smallequation}
\label{eq:ability-vs-cover-alltraces}
\Ability(\mathcal{G};\mathcal{F})
\ge
\max_{j\in[M]}
\left\{
\rho_j
\left(
1-\sum_{t=1}^{m_j}
\frac{\binom{n-s_{j,t}}{c\cdot n}}{\binom{n}{c\cdot n}}
\right)
\right\}.
\end{smallequation}

\noindent
\textbf{Interpretation.}
The lower bound is monotone in the solve-level coverage
$c=\Cover(\mathcal{L};\mathcal{L}^\star)$.
However, the support-set sizes $s_{j,t}$ depend on the target proof-skeleton structure, so the bound cannot in general be reduced to a closed-form expression purely in terms of $c$.

When analyzing protocol-level coverage over a large lemma set, it is natural to have
$n=|S^\star|\gg s_{j,t}$.
In this regime, with $k=cn$, the combinatorial term admits the approximation
\begin{smallequation}
\label{eq:comb-approx-large-n}
\frac{\binom{n-s_{j,t}}{k}}{\binom{n}{k}}
\approx
\left(1-\frac{k}{n}\right)^{s_{j,t}}
=
(1-c)^{s_{j,t}}.
\end{smallequation}
Substituting~\eqref{eq:comb-approx-large-n} into Eq.~\eqref{eq:ability-vs-cover-alltraces} yields the following coverage-driven shape:
\begin{smallequation}
\label{eq:ability-vs-cover-alltraces-coverage-driven}
\Ability(\mathcal{G};\mathcal{F})
\gtrsim
\max_{j\in[M]}
\left\{
\rho_j
\left(
1-\sum_{t=1}^{m_j}(1-c)^{s_{j,t}}
\right)
\right\}.
\end{smallequation}
This approximation suggests that capability increases with coverage in a concave manner through the terms
$1-(1-c)^{s_{j,t}}$, and that larger support sets lead to a faster rise as $c$ increases.

\section{Experiment}
\label{sec:experiment}

\subsection{Dataset}
\label{sec:dataset}
% 中文直译：我们的基准数据集来源于 AutoSM（ICSE'25），其中包含协议的形式化模型与对应的 Tamarin lemmas；我们使用其 Tamarin 部分作为评估对象，并将仓库链接放在脚注中以便复现。
Our dataset is derived from \textsc{AutoSM}\footnote{\url{https://github.com/zerrymore/AutoSM}} \cite{10.1109/ICSE55347.2025.00197}, which include 18 protocols and 90 human-expert-written lemmas from different cryptography scenarios. We have added 54 mutation protocols based on this dataset, which will be explained further in section \ref{sec:results}.

% 中文直译：对每个协议，数据集给出一组专家 lemma 来刻画安全属性（如 secrecy、authentication、sanity）；在后续实验中我们把这些 reference lemmas 作为目标属性定义，用以评估模型生成 lemmas 的质量。
For each protocol theory $P$, the dataset includes a reference lemma set that specifies intended security properties. The properties can be classified by secrecy, authentication, and sanity.

\subsection{Results}
\label{sec:results}
To answer the two research questions raised in the introduction (section \ref{sec:introduction}), we organize the evaluation around the following five sub-questions.

\textbf{RQ1.} How well do LLMs produce syntactically valid specifications in CSPV? 

\textbf{RQ2.} How effective is CRoST as an evaluation metric? 

\textbf{RQ3.} What are the capability boundaries of LLMs in CSPV? 

\textbf{RQ4.} How do LLMs perform across different task types in CSPV? 

\textbf{RQ5.} What is the computational cost of using LLMs in CSPV?

% RQ1：LLM在CSPV的语法正确表现如何？
% RQ2：CRoST的测评方法是否有效？
% RQ3：在CSPV中，LLM的能力边界是什么？
% RQ4：在不同类型的任务下，LLM的表现如何？
% RQ5：LLM在CSPV中开销有多大？
% RQ1: How well do LLMs produce syntactically valid specifications in CSPV?
% RQ2: How effective is CRoST as an evaluation metric?
% RQ3: What are the capability boundaries of LLMs in CSPV?
% RQ4: How do LLMs perform across different task types in CSPV?
% RQ5: What is the computational cost of using LLMs in CSPV?
\begin{rqbox}
\textbf{RQ1.} How well do LLMs produce syntactically valid specifications in CSPV?
\end{rqbox}
% \subsection{Few-shot learning}
% \label{subsec:fewshot}
% 在进行测试前，首先需要考虑的一个问题是相比于python等语言，tamarin的语言是一种可供学习样本较少的领域专用语言（DSL），为了发挥大模型对密码协议的推理与安全属性定义能力，我们需要首先使大模型学习该DSL知识
% 幸运的是，由于tamarin的lemma语法较为简单，见ref{fig:tamarin-secrecy-lemmas}，其所有组成仅为一阶逻辑谓词与时序关系，我们认为可以通过简单的few-shot learning的方式使大模型生成可用的lemma，样本为协议的rule内容，输出为1-4个对应的lemma，我们通过简单的对比实验说明这点，图结果显示在输入rule直接提问lemma的情况下lemma通过tamarin编译的正确率仅为25.7%，可以认为qwen3-coder模型基本无tamarin相关语法知识，在使用few-shot之后提高到了56.1%。在表中，使用few-shot方法，大部分的模型正确率都在70以上。

% Tamarin 的规范语言是一种小众的领域专用语言（DSL），并非通用编程语言例如Python、C、Java等，LLM的DSL生成任务通常面临难以有效生成正确语法的问题，因此我们首先探究LLM能否有效生成可用的lemma描述语言。
% 这里找一下更好地表达
% Since Tamarin uses a domain-specific specification language rather than a general-purpose programming language such as Python, C, or Java, relevant training data are relatively scarce \cite{tamarin-manual,10.1007/978-3-642-39799-8_48, tuccio-etal-2025-grammar, jiao-etal-2024-exploring}. LLMs often struggle to generate DSL code with correct syntax \cite{11334235}. We first investigate whether LLMs can reliably produce usable lemma specifications.
% We therefore adopt few-shot prompting (in-context learning) \cite{10.5555/3495724.3495883, min-etal-2022-rethinking} by providing a small number of demonstrations in the form of (protocol-rule snippet →\rightarrow 1--4 corresponding lemmas).

Since Tamarin uses a domain-specific specification language rather than a general-purpose programming language such as Python, C, or Java, publicly available examples are relatively limited \cite{tamarin-manual,10.1007/978-3-642-39799-8_48, tuccio-etal-2025-grammar, jiao-etal-2024-exploring}. In practice, this challenge often appears first at the syntactic level. LLMs may fail to produce lemma specifications that are well-formed and tool-checkable \cite{11334235}. 
We therefore begin by investigating whether LLMs can reliably produce usable lemma specifications, and whether few-shot prompting (in-context learning) \cite{10.5555/3495724.3495883, min-etal-2022-rethinking} can improve syntactic validity and downstream verification performance by providing a small number of demonstrations in the form of multiset rewriting rules with 1--4 corresponding lemmas.

% 中文直译：我们首先做一个最直接的对比实验：对同一批 rule 输入，分别使用 zero-shot 与 few-shot 提示，让模型生成 lemmas，并以“能否通过 Tamarin 解析/编译（无语法错误）”作为 Syntax Success 的判定标准。图ref{fig:zeroshot_vs_fewshot_syntax}展示了 qwen3-coder-plus 的结果：few-shot 方法提升了 30.4 个百分点；这表明在缺少 DSL 先验时，few-shot 对“先学会写对”非常关键。
We start with a comparison between zero-shot and few-shot prompting on the same set of rule inputs.
The \emph{syntax success} means that the generated lemmas can be parsed and compiled by Tamarin without syntax errors.
Figure~\ref{fig:zeroshot_vs_fewshot_syntax} shows the results for qwen3-coder-plus.
Compared with zero-shot prompting, few-shot prompting improves the syntax success rate by 30.4 percentage points. This indicates that few-shot prompting can significantly improve the syntax success rate.

\begin{figure}[htb]
  \centering
  \includegraphics[width=\linewidth]{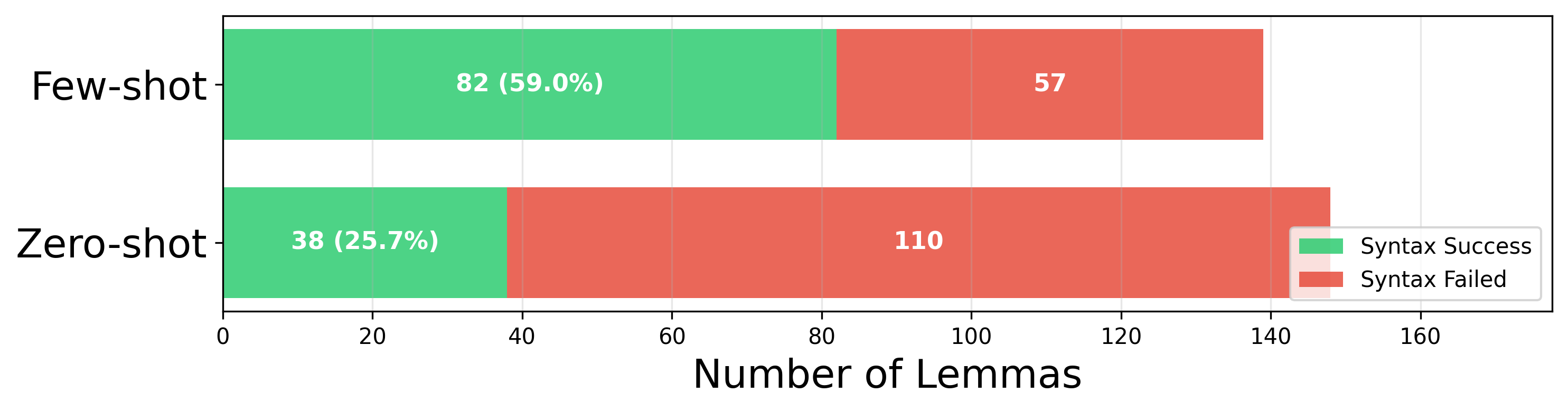}
  \caption{\textbf{Zero-shot vs.\ few-shot syntax success for \texttt{qwen3-coder-plus}.}}
  \label{fig:zeroshot_vs_fewshot_syntax}
\end{figure}

% 中文直译：并在表ref{tab:syntax_rate} 中汇总各模型的语法通过率；大多数模型在 few-shot 下超过 70\%，从而为后续基于 CRoST 的语义/能力评估提供了可用的输入前提。在后续的所有主实验中，我们默认使用 few-shot 提示作为生成配置
Table~\ref{tab:syntax_rate} summarizes the syntax success rates across models; the majority exceed 70\% under few-shot prompting, providing a necessary foundation for the later evaluation that focuses on proof-search behavior and security property coverage rather than surface syntax.
In all subsequent experiments, we use few-shot prompting as the default generation setting.

\begin{table}[htb]
\centering
\caption{Syntax Success Rate by Model.}
\label{tab:syntax_rate}
\small
\setlength{\tabcolsep}{4pt}
\begin{tabular}{lccc}
\toprule
Model & Generated & Syntax Success & Rate (\%) \\
\midrule
grok-4 & 75 & 74 & 98.67 \\
claude-sonnet-4.5 & 245 & 241 & 98.37 \\
gpt-5.2-codex & 118 & 114 & 96.61 \\
claude-opus-4.5 & 166 & 155 & 93.37 \\
gemini-3-flash-preview & 95 & 86 & 90.53 \\
gpt-5.2 & 204 & 169 & 82.84 \\
claude-haiku-4.5 & 290 & 230 & 79.31 \\
qwen-plus & 78 & 57 & 73.08 \\
deepseek-v3.2 & 230 & 166 & 72.17 \\
gpt-4 & 74 & 46 & 62.16 \\
qwen3.5-plus & 218 & 135 & 61.93 \\
qwen3-coder-plus & 139 & 82 & 58.99 \\
\bottomrule
\end{tabular}
\end{table}

\begin{rqbox}
\textbf{RQ2.} How effective is CRoST as an evaluation metric?
\end{rqbox}

% 在第三节中，我们证明了CRoST与lemma生成能力之间的关系。此外，我们设计了突变实验来进一步证明CRoST的合理性。
In section \ref{sec:methodology}, we prove the relationship between CRoST and security properties similarity. Furthermore, we designed mutation experiments \cite{10.1007/978-3-031-78946-5_20} to further prove the rationality of CRoST.

% % 由于真实密码协议中漏洞稀缺，为了证明CRoST指标与lemma生成能力之间的关系，我们人为设计了一些明显的漏洞注入到密码协议中，漏洞类型如表所示，并在同一生成lemma上观察结果是否发生翻转，即，如图所示，若某条 lemma 在原始协议上被验证（verified），但在注入漏洞后的变异协议上被反例推翻（falsified），则认为该漏洞被该 lemma 检出。
% We apply the same operator suite across protocols to generate vulnerable variants that capture common classes of symbolic protocol bugs, including confidentiality breaks, authentication breaks, and executability/sanity breaks \cite{jia2011mutation,xmen2020,xmen2023}.
% For each protocol, we instantiate mutants by applying one operator at one applicable program point, while keeping the lemma set unchanged.

% 中文直译：由于真实密码协议中漏洞稀缺。为验证 CRoST 是否能反映 lemma 的检验能力，我们采用突变测试：人为向协议规则注入一组显式漏洞（见表ref{tab:mutation_ops}），并在同一组生成 lemmas 上观察验证结果是否发生翻转。
% 中文直译：如ref{fig:mutation_example} 所示，若某条 lemma 在原始协议上被验证（verified），但在注入漏洞后的变异协议上被反例推翻（falsified），则认为该漏洞被该 lemma 检出；我们用翻转率刻画在一组变异体上的整体检出能力，并据此检验 CRoST 的有效性。
% 这里需要加一段，确实和CRoST一样mutation test的确也可以作为一种测试lemma生成能力的指标，但是突变测试的局限性在于需要手工设计协议漏洞，需要大量的工作量，此外，真实协议的漏洞是隐蔽的，人工设计突变样本难以cover全部的lemma能力，难以有效观察形式化验证的搜索过程；因此我们在这里仅用突变测试印证CRoST的rationality
Since real-world cryptographic protocol vulnerabilities are scarce, we adopt mutation testing to assess whether \textsc{CRoST} is a rational metric \cite{9793974}. Specifically, we inject a suite of simple, explicit flaws into protocol rules (Table~\ref{tab:mutation_ops}), generated with Gemini-3-Pro and manually vetted by human experts. These vulnerabilities are designed as controlled, verification-relevant fault proxies. Since security properties characterize the intended security behavior of a protocol, any vulnerability should violate at least one of them. Therefore, if a lemma captures the target property well, a local perturbation that breaks that property should cause the verification result to flip, allowing us to test whether \textsc{CRoST} responds appropriately to property-breaking changes.
As illustrated in Figure~\ref{fig:mutation_example}, if a lemma is \emph{verified} on the original protocol but becomes \emph{falsified} on a mutant, we treat the injected vulnerability as detected by that lemma.
We then use the resulting flip rate over mutants to assess the validity of \textsc{CRoST}.

% 中文注释：需要说明的是，突变测试本身也可作为衡量 lemma 生成能力的指标，但它存在局限：构造高质量的协议漏洞需要大量人工设计工作；且真实协议漏洞往往隐蔽多样，人工突变样本难以覆盖完整的能力谱系，也难以呈现形式化验证过程中的细粒度搜索行为。因此我们在此仅将突变测试作为对 CRoST 合理性的补充印证，而非主要评测指标。
What needs to be explained here is mutation testing itself can serve as an indicator of lemma-generation capability, but it has practical limitations \cite{10.1145/2635868.2635929, 8453121}.
Constructing protocol flaws requires substantial manual effort, and real-world vulnerabilities are often covert, making handcrafted mutants unlikely to cover the full spectrum of lemma capability.
Moreover, flip-based outcomes provide limited visibility into the verifier’s fine-grained proof search behavior.
Therefore, we use mutation testing here only as supporting evidence for the rationality of \textsc{CRoST}, rather than as our primary evaluation metric \cite{10.1145/3324884.3416667}.

\begin{figure}[htb]
  \centering
  \includegraphics[width=\linewidth]{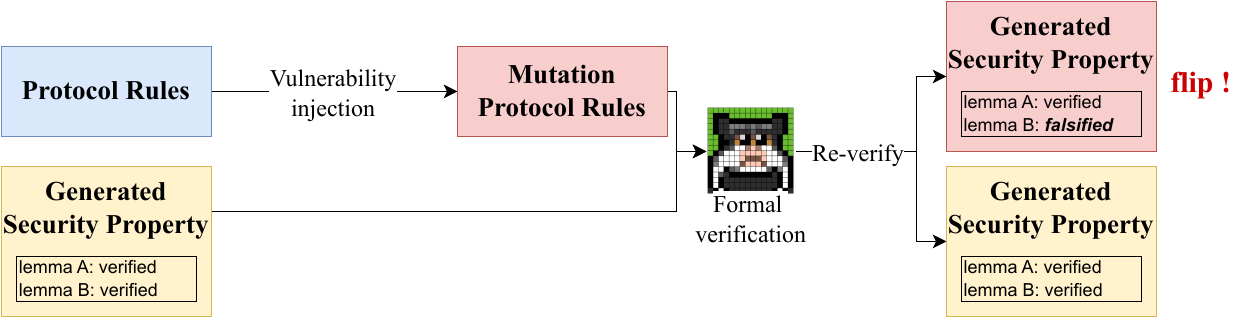}
  \caption{\textbf{Mutation test experiment sketch.}}
  \label{fig:mutation_example}
\end{figure}

\begin{table}[htb]
\centering
\caption{Simple modifications to construct vulnerable protocol variants.}
\label{tab:mutation_ops}
\small
\setlength{\tabcolsep}{4pt}
\begin{tabular}{p{0.3\linewidth}p{0.65\linewidth}}
\toprule
\textbf{Type} & \textbf{Operation description} \\
\midrule
Direct leakage &
Add \texttt{Out(...)} to directly release a secret value. \\

Weak key derivation &
Modify key-derivation inputs to include public constants, making the derived key predictable/computable by the adversary. \\

Removed verification &
Delete signature/MAC verification guards (e.g., dropping an \texttt{Eq(verify(...), true)}-style check), enabling message forgery. \\

Weakened binding &
Break identity/nonce binding by replacing a bound variable with an unconstrained one. \\

Missing freshness check &
Remove freshness-related checks so that replayed messages can be accepted without detection. \\

Incomplete message &
Remove authentication-critical fields (e.g., a nonce) from transmitted messages, preventing the receiver from validating sanity/freshness. \\

Broken message flow &
Change message tags/formats so expected parsing/transition rules no longer match, disrupting sanity. \\

Missing response &
Directly delete a required output/response. \\
\bottomrule
\end{tabular}
\end{table}

% 中文直译ref{fig:flip_rate_analysis} 给出覆盖率与翻转率的散点关系，并报告皮尔逊相关系数；整体相关系数为 0.796，且在 secrecy/authentication/sanity 三类属性上分别为 0.670/0.769/0.681，表明 CRoST 与漏洞检出能力具有稳定的正相关。
We tested lemmas generated by current state-of-the-art (SOTA) LLMs. Figure~\ref{fig:flip_rate_analysis} shows a strong positive correlation between CRoST coverage and flip rate (correlation coefficient = 0.75, p-value = 0.005).
This evidence supports the validity of CRoST as a fine-grained, tool-grounded indicator for how well generated lemmas capture security properties. 

\begin{figure}[htb]
  \centering
  \includegraphics[width=0.8\linewidth]{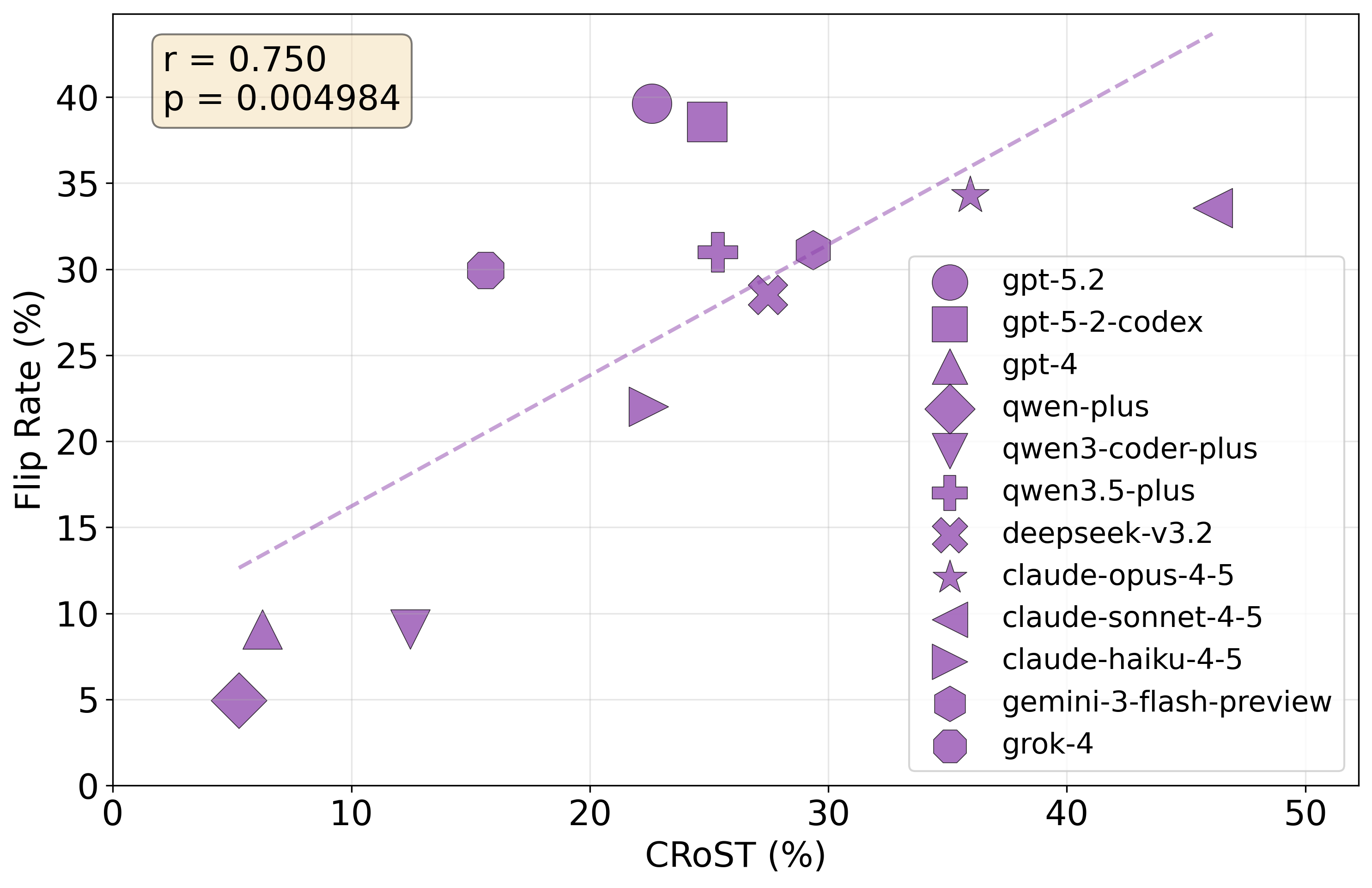}
  \caption{\textbf{CRoST coverage vs.\ lemma flip rate under mutation testing.}}
  \label{fig:flip_rate_analysis}
\end{figure}

\begin{rqbox}
    \textbf{RQ3.} What are the capability boundaries of LLMs in CSPV?
\end{rqbox}
% 这里还需要说明的一点是：其实想证伪这个问题：如果这是一个穷举搜索的问题，那么是否让一个能力较弱的大模型运行足够多次，就能覆盖所有的路径？-实际不行，会达到一个上界
% 中文直译：在 RQ2 已验证 CRoST 作为能力度量有效之后，本节进一步关注“模型能力边界”本身：不同模型的 CRoST 差异究竟来自何处？一个直观的可能性是，更强的模型只是生成了更多 lemmas，从而更容易覆盖到目标 solve-tree 的路径；若如此，能力边界可能主要受限于采样规模而非推理质量。
With the validity of \textsc{CRoST} established in RQ2, we next examine the capability boundary of LLMs in CSPV.
A natural hypothesis is that higher \textsc{CRoST} scores may be driven mainly by generating more lemmas, i.e., stronger models simply explore more candidates and thus cover more solve-tree paths \cite{Kaplan2020ScalingLF}.

% 中文直译：图fig:lemma_count_vs_coverage} 检验了“生成 lemma 数量”与 CRoST 的关系。结果显示二者不是很相关（r=0.580r=0.580, p=0.048p=0.048），且离散性明显：在相近的生成规模下，不同模型的 CRoST 可能相差较大；同时，生成更多 lemmas 并不必然带来更高覆盖率。这表明 CSPV 中的能力差异不能简单归因于“生成数量”。
Figure~\ref{fig:lemma_count_vs_coverage} evaluates the relationship between the number of generated lemmas and \textsc{CRoST}.
We observe less correlation with substantial dispersion ($\text{p-value}=0.04825$).
This suggests that CSPV capability differences are not explained primarily by lemma quantity.

\begin{figure}[htb]
  \centering
  \includegraphics[width=0.8\linewidth]{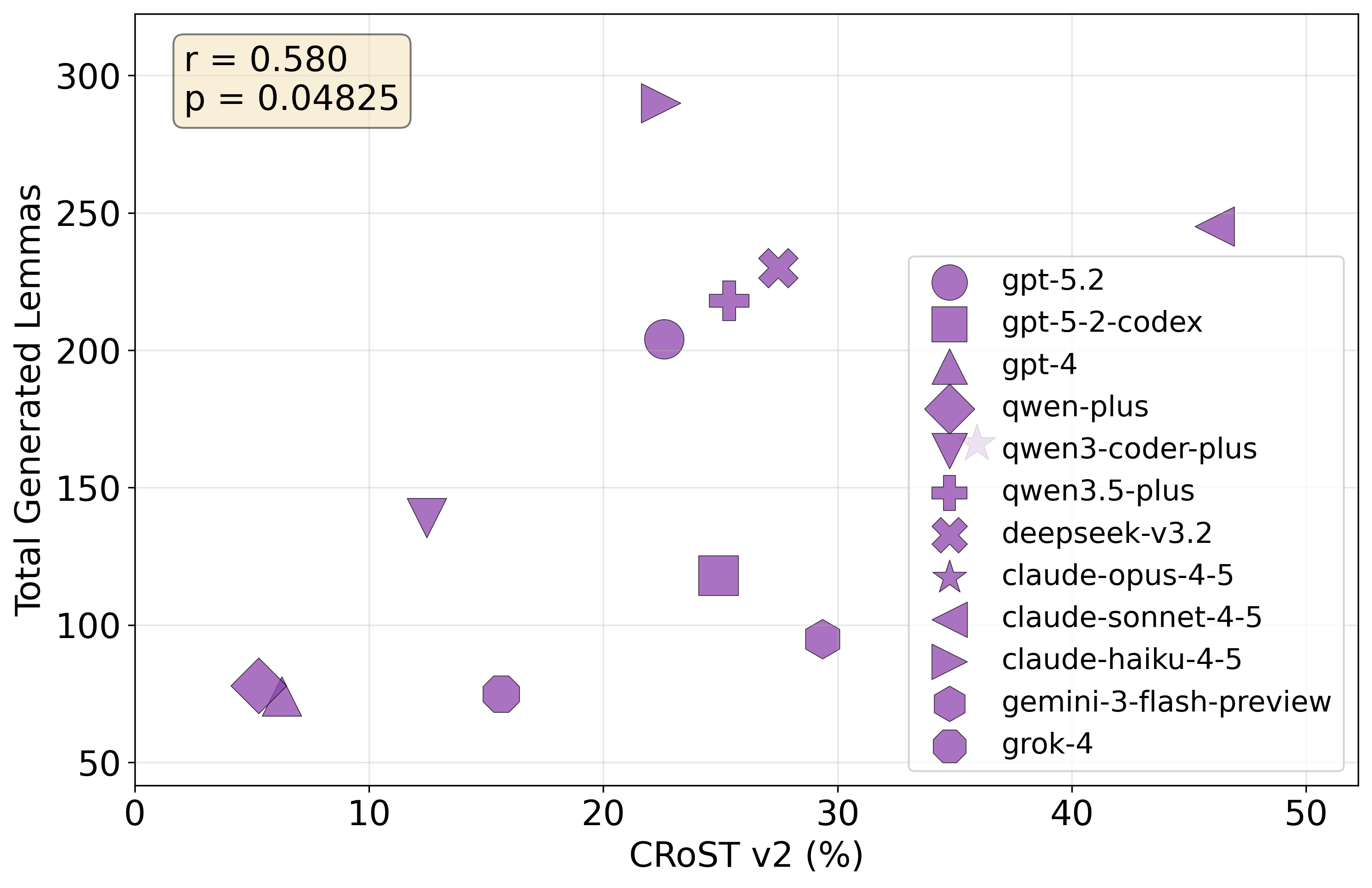}
  \caption{\textbf{Lemma count vs.\ \textsc{CRoST}.}
  Each point is an LLM.}
  \label{fig:lemma_count_vs_coverage}
\end{figure}

% 中文直译：进一步地，我们固定同一模型，重复运行完整 pipeline，并统计“累计生成规模”与“累计覆盖率”的增长趋势（fig:cumulative_coverage}）。在 12 轮迭代中，累计生成 lemma 数从 239 增长至 2672，但 CRoST 仅提升约 24 个百分点，并在第 8 轮后基本停止增长。这反映出明显的收益递减与覆盖率饱和：继续采样主要带来重复或弱变体，而非新增可覆盖的 solve 路径。因而，提升 CSPV 表现的关键在于提升 lemma 的有效性与覆盖增量（能触达新 solve 路径），而不是追求更高的生成数量。
To further probe the boundary under repeated sampling, we run the same pipeline multiple times with a fixed LLM deepseek-v3.2 and track cumulative statistics (Figure~\ref{fig:cumulative_coverage}).
Across 12 iterations, cumulative generated lemmas increase from 239 to 2672, while \textsc{CRoST} improves by only about 24 percentage points and largely plateaus after the 8th run.
This indicates diminishing returns and coverage saturation: additional sampling mainly produces redundant or weak variants that contribute little new solve-path coverage.
Therefore, improving CSPV performance requires generating more informative lemmas that contribute incremental solve-path coverage, rather than increasing the raw number of candidates.

\begin{figure}[htb]
  \centering
  \includegraphics[width=0.8\linewidth]{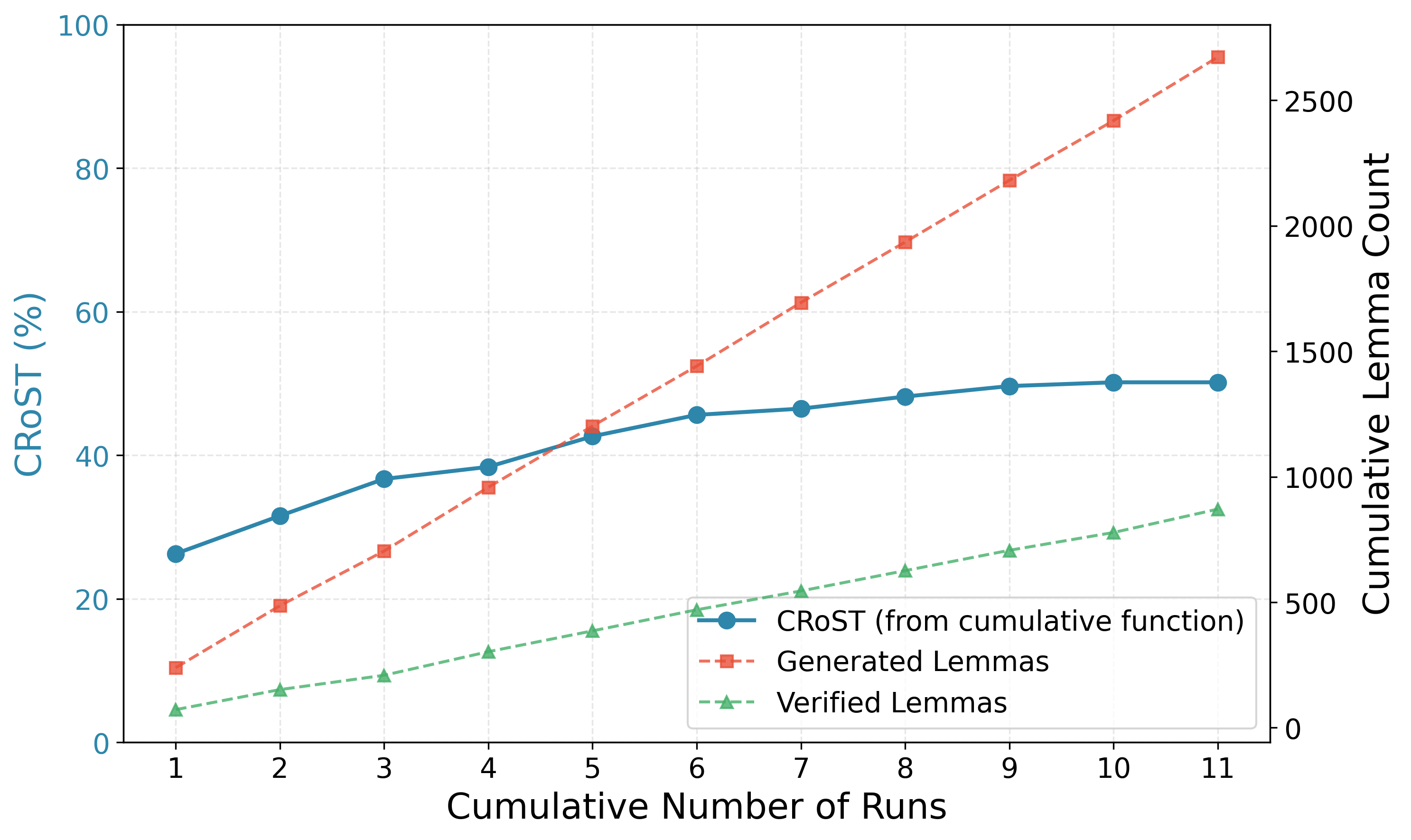}
  \caption{\textbf{Cumulative sampling and \textsc{CRoST}.}}
  \label{fig:cumulative_coverage}
\end{figure}

\begin{rqbox}
    \textbf{RQ4.} How do LLMs perform across different task types in CSPV?
\end{rqbox}

% 中文直译：我们将不同类型的安全属性视为不同的 CSPV 任务类型（secrecy / authentication / sanity），并分别报告各模型在这些任务上的 CRoST，以刻画其能力结构与偏差。
We treat major security-attribute families as distinct task types in CSPV (secrecy, authentication, and sanity), and report per-task \textsc{CRoST} to characterize capability differences and model biases across tasks.

The results in Table~\ref{tab:CRoST_by_property} indicate that, on average, SOTA LLMs perform better on authentication and sanity properties than on secrecy properties. Figure~\ref{fig:lemma_crost_top4} offers a different, lemma-level perspective. The highest-scoring target lemmas are disproportionately concentrated in sanity properties. A plausible explanation is that secrecy and authentication properties are often expressed as \emph{all-traces} universally quantified lemmas, whereas sanity properties are more commonly formulated as \emph{exists-trace} existential lemmas. This suggests that current LLMs are generally more effective at generating existential security properties than universally quantified ones.

\begin{table}[htb]
\centering
\caption{CRoST by Security Property (\%).}
\label{tab:CRoST_by_property}
\small
\setlength{\tabcolsep}{4pt}
\begin{tabular}{lcccc}
\toprule
Model & Overall & Secrecy & Authentication & Sanity \\
\midrule
claude-sonnet-4.5 & 46.11 & 30.91 & 55.85 & 39.97 \\
claude-opus-4.5 & 35.94 & 42.10 & 20.02 & 39.83 \\
gemini-3-flash-preview & 29.36 & 35.78 & 24.28 & 30.74 \\
deepseek-v3.2 & 27.46 & 20.51 & 24.25 & 35.47 \\
\textit{top4 LLMs} & \textit{38.82} & \textit{33.90} & \textit{40.56} & \textit{39.95} \\
qwen3.5-plus & 25.37 & 9.32 & 36.48 & 28.34 \\
gpt-5.2-codex & 24.91 & 27.50 & 21.55 & 20.78 \\
gpt-5.2 & 22.60 & 14.87 & 25.98 & 27.44 \\
claude-haiku-4.5 & 22.46 & 8.44 & 19.90 & 35.89 \\
grok-4 & 15.63 & 14.87 & 7.32 & 22.85 \\
qwen3-coder-plus & 12.46 & 7.07 & 21.24 & 5.37 \\
gpt-4 & 6.28 & 6.67 & 3.24 & 3.54 \\
qwen-plus & 5.28 & 2.54 & 6.83 & 6.40 \\
\bottomrule
\end{tabular}
\end{table}

\begin{figure*}[t]
    \centering
    \includegraphics[width=\linewidth]{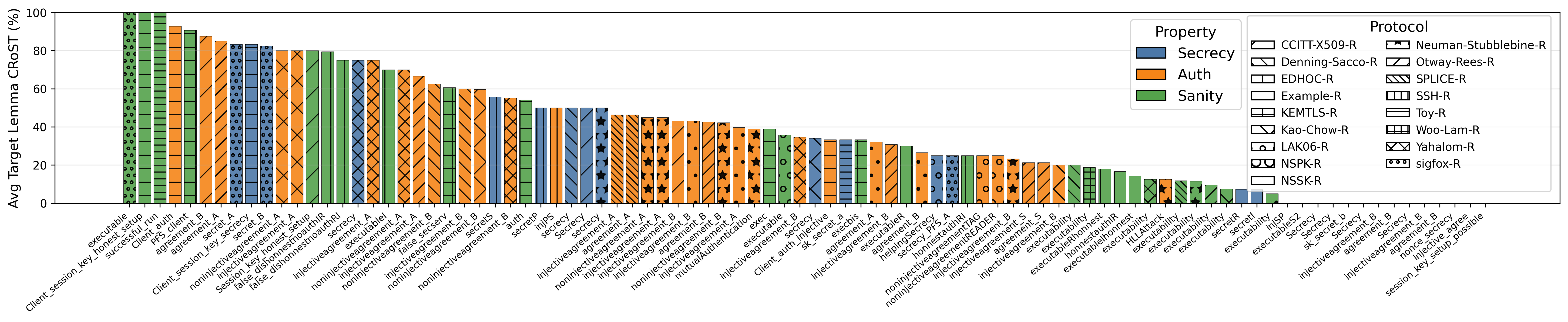}
    \caption{Top4 LLMs' \textsc{CRoST} on target lemma.}
    \label{fig:lemma_crost_top4}
\end{figure*}

% 中文直译：为理解上述差异的具体原因，我们对 Overall 表现最好的四个模型进一步开展 case study（\emph{claude-sonnet-4.5}、\emph{claude-opus-4.5}、\emph{gemini-3-flash-preview}、\emph{deepseek-v3.2}），重点分析大型协议中的失败模式。
To understand the failure modes of SOTA LLMs, we conduct a focused case study on the top4 LLMs by overall \textsc{CRoST} (\emph{claude-sonnet-4.5}, \emph{claude-opus-4.5}, \emph{gemini-3-flash-preview}, and \emph{deepseek-v3.2}), and analyze representative failures on lemma generation.

\begin{casebox}
    \textbf{Case study.} Stage point confusion in large protocols.
\end{casebox}
% 需要说明对于rule作用的理解能力
% 需要加一张示意图
% 为了进一步阐述LLM生成lemma的失败模式，我们以多阶段 AKE 协议 \texttt{SSH-R} 为例进行 case study。该协议的目标属性绑定到多个阶段点（不同的 \texttt{Accept*} 事件）。
To concretely illustrate failure modes of LLMs to generate lemma, we present a case study on the multi-phase Authenticated Key Exchange (AKE) protocol \texttt{SSH-R}.
Its target properties are bound to stage-specific action facts.

% 在该协议中，客户端的认证由两个阶段构成:第一阶段（action fact AcceptP）对应“传输层密钥交换 + 服务器认证”;第二阶段（action fact AcceptP2）对应“用户认证 + 最终密钥确认”。\texttt{injPS} 约束第一阶段的 \texttt{AcceptP} 必须在此前由对端 \texttt{AcceptS} 见证；\texttt{injSP} 约束第二阶段的 \texttt{AcceptS2} 必须对应此前的 \texttt{AcceptP2}；\texttt{executableS2} 的可执行性见证要求最终阶段事件 \texttt{AcceptS2} 实际发生。
In \texttt{SSH-R}, client-side authentication proceeds in two stages: the first stage is marked by the action fact \texttt{AcceptP} and corresponds to transport-layer key exchange plus server authentication, while the second stage is marked by \texttt{AcceptP2} and corresponds to \emph{user authentication plus final key confirmation}.
Accordingly, the target lemmas are stage-specific rather than generic agreement templates:
\texttt{AcceptP} (phase~1) is bound to a prior \texttt{AcceptS} witness,
\texttt{AcceptS2} (phase~2) is bound to a prior \texttt{AcceptP2}.

We observed that the generated lemmas still demonstrate a reasonable understanding of the roles of action facts.
That is, LLMs are often able to distinguish whether an action fact corresponds to key establishment, authentication, or message confirmation, and thus tend to preserve the high-level functional intent of the protocol. More importantly, however, LLMs also exhibit a representative form of \emph{stage-point confusion} in generated lemmas.
In Figure~\ref{fig:failure_mode}, while the LLM captures the high-level intent of agreement/authentication, it misaligns stage points by relating \texttt{AcceptP2} (phase~2) to \texttt{AcceptS} (phase~1), instead of the target \texttt{AcceptP} (phase~1) to \texttt{AcceptS} (phase~1), and forgets the temporal relationship.
This phase mismatch breaks the intended correspondence structure.

\begin{figure}[htb]
    \centering
    \includegraphics[width=0.8\linewidth]{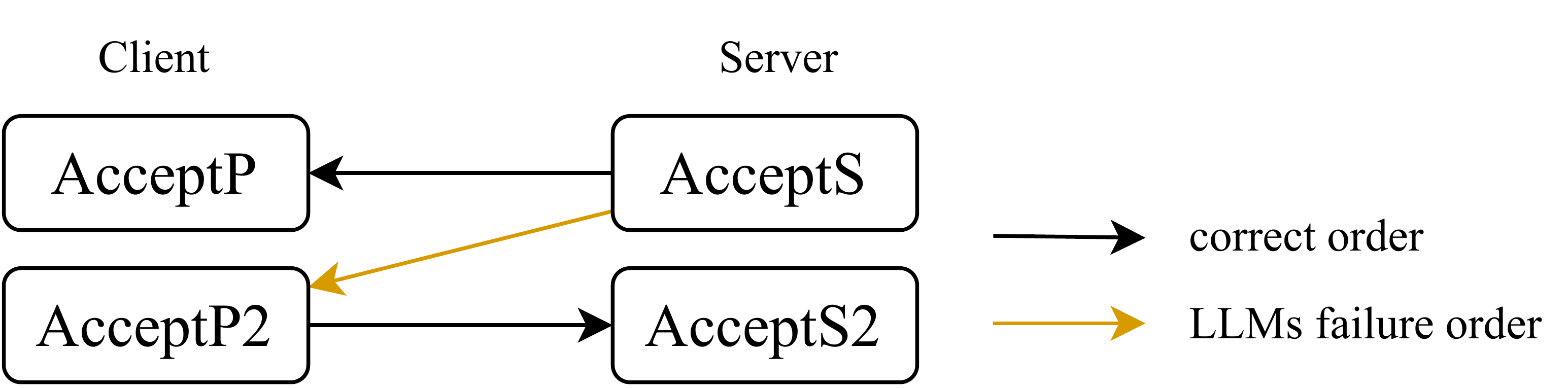}
    \caption{Example of an LLM failure mode in the basic SSH protocol.}
    \label{fig:failure_mode}
\end{figure}

% 中文直译：这种“阶段点混淆”在大型协议中尤为致命：由于协议由多轮交互与多种密码原语组合而成，往往引入大量阶段性 action facts（例如 \texttt{Accept}, \texttt{Accept2}, \texttt{AcceptP}, \texttt{AcceptP2} 等）。一旦在 lemma 中混用了不同阶段点，即使只偏移一个谓词名，也会使证明搜索无法命中目标 lemmas 的关键 solve 路径，从而造成覆盖率在该属性族上显著下降（例如 authentication 只能停留在较低水平）。
% 该 case study 表明：当前 LLM 无法精确对齐阶段性事件与其时序约束，因此在CSPV任务上难以处理较为复杂的协议。
This issue is particularly damaging for large protocols with many phase-specific action facts induced by multi-round message flows and composed cryptographic primitives.
A single stage point mix-up can prevent proof search from reaching the target solve-tree paths, resulting in substantial coverage loss for the affected lemma family.
This case study indicates that current LLMs often fail to precisely align stage-specific events with their temporal constraints, which makes them struggle with CSPV tasks on more complex protocols.

\begin{rqbox}
    \textbf{RQ5.} What is the computational cost of using LLMs in CSPV?
\end{rqbox}

% 中文直译：我们从两方面衡量开销：（1）LLM 生成阶段的 token 输出规模；（2）对生成 lemmas 进行形式化验证的运行时间。所有实验在同一硬件与同一 Tamarin 配置下进行：Intel(R) Core(TM) i9-14900 (2.00 GHz), 24 cores, 32GB RAM；Tamarin 运行参数为 \texttt{+RTS -N -RTS --auto-sources}（启用多核与自动 sources 预处理）。
We measure cost from two aspects: (i) LLM token output during lemma generation, and (ii) runtime for formally verifying generated lemmas.
All experiments were executed on Intel(R) Core(TM) i9-14900 (2.00 GHz), 24 cores, 32GB RAM, with Tamarin configured as \texttt{+RTS -N -RTS --auto-sources} \cite{tamarin-manual}.

% 中文直译：Token 统计方面，我们统一使用 Python 的 \texttt{tiktoken} 库计算 token 数，以便在不同模型间保持可比性；所有调用的输入提示（prompt）长度固定为 1049 tokens（仅报告输出 token）。需要注意，不同厂商模型的原生 tokenizer 可能不同，因此该统计用于横向对比而非精确计费。
For token accounting, we use the Python \texttt{tiktoken} library for consistent counting across LLMs.
The prompt length is fixed at 1049 tokens for all calls; Table~\ref{tab:token_cost} reports average output tokens only.
% Since providers may use different native tokenizers, these counts are intended for comparability rather than exact billing.

\begin{table}[t]
\centering
\caption{Average output tokens per call.}
\label{tab:token_cost}
\small
\setlength{\tabcolsep}{4pt}
\begin{tabular}{lc}
\toprule
LLM & Avg Output Tokens \\
\midrule
claude-sonnet-4.5 & 1143 \\
claude-opus-4.5 & 691 \\
gemini-3-flash-preview & 455 \\
deepseek-v3.2 & 844 \\
\bottomrule
\end{tabular}
\end{table}

% 中文直译：时间开销方面，我们报告“验证生成 lemmas”的总运行时间，并与原始协议 target lemmas 的验证时间（baseline）对比。复杂协议的形式化分析通常会消耗大量时间与内存（已有工作在真实协议模型上报告了小时级运行时间与数十 GB 内存占用），因此验证阶段的时间开销是 CSPV 可用性的关键因素。
For runtime, we report the total time spent verifying generated lemmas, and compare it against verifying the original target lemmas as a baseline.
Formal analysis of realistic protocol LLMs can be time- and memory-intensive, making verification runtime a key usability factor for CSPV.

\begin{table}[htb]
\centering
\caption{Verification runtime of generated lemmas.}
\label{tab:runtime_cost}
\small
\setlength{\tabcolsep}{3pt}
\begin{tabular}{lccc}
\toprule
Setting & Lemmas & Total Time (s) & Mean (s/lemma) \\
\midrule
Baseline (target lemmas) & 91 & 59.91 & 0.66 \\
claude-sonnet-4.5 & 245 & 1267.76 & 5.17 \\
claude-opus-4.5 & 166 & 1127.33 & 6.79 \\
gemini-3-flash-preview & 95 & 238.24 & 2.51 \\
deepseek-v3.2 & 230 & 1151.46 & 5.01 \\
\bottomrule
\end{tabular}
\end{table}

% 中文直译：相较 baseline，验证生成 lemmas 的总时间约为 4--21×\times（分别约为 4.0×\times/21.2×\times 的范围）。这表明在 CSPV 场景中，生成阶段的开销之外，\emph{大量错误或弱质量 lemmas 带来的验证尝试成本}是主要瓶颈之一；因此，构造更高质量、能减少无效尝试的 lemmas（例如提升有效覆盖增量）是重要的后续方向。
Compared to the baseline, verifying generated lemmas incurs roughly a $4\times-21\times$ slowdown.
This suggests that, beyond generation cost, the dominant overhead in CSPV can stem from verification attempts on incorrect or low-utility lemmas.
Improving lemma quality to reduce wasted verification effort (e.g., by increasing effective coverage gains) is therefore an important direction.
The mean per-lemma verification time also increases from $0.66$s to $2.51$--$6.79$s, indicating that generated lemmas also tend to induce longer formal proof-search cost per lemma.

\subsection{Threats to Validity}
\label{subsec:threats}

\textbf{External validity.}
Our benchmark covers 18 protocols and three major property families: secrecy, authentication, and sanity.
However, protocol complexity and property distribution may affect the observed results, especially in RQ4 where we compare model performance across different security-property types.
Therefore, the results may not fully generalize to larger protocol collections or to property families not covered by our benchmark.

\textbf{Internal validity.}
LLM-based lemma generation is sensitive to prompt design, in-context examples, and decoding randomness.
Although we use a unified prompting strategy and the same evaluation pipeline across models, the inherent randomness of LLM generation may still introduce variability in the generated lemmas and the resulting CRoST scores.

\section{Discussion}
\label{sec:discussion}

% \subsection{Limitation}
% \label{subsec:limitation}

% Several factors may influence the observed results.

% % 数据集可能不均衡，协议的复杂程度、数据集的规模可能会影响实验结果，特别是在RQ4中的不同安全属性类型的表现
% % 尽管使用相同的提示词，大模型的回答仍然会受到随机的生成过程的影响，可能会导致结果的波动。
% % CRoST
% \textbf{Dataset imbalance.}
% The complexity of protocols and the size of the dataset may affect the results, especially in RQ4 where the performance across different security property types is analyzed.

% \textbf{Prompt sensitivity.}
% The generation quality is affected by prompt design and in-context examples. 
% Although we use a unified prompting strategy, the inherent randomness in LLMs generation can lead to variability in results.

% \textbf{Metric limitations.}
% \textsc{CRoST} may have limitations. When a proof path contains simple solve items, such as \texttt{solve(!KU(x))}, a node in an unrelated generated lemma may still cover the target lemma at that point, leading to a small number of false positives. This issue can be mitigated by introducing additional constraints into the \textsc{CRoST} computation.

\subsection{Future work}
% 未来工作可沿三个更具体的方向推进。首先，CRoST 提供了比“proved/falsified”更细粒度的可学习信号，因此可将其作为强化学习（或其它优化方法）的奖励/反馈，用于直接优化 lemma 生成策略以最大化覆盖增量并减少无效验证尝试。其次，要扩展到复杂大型协议，关键瓶颈在于阶段点混淆：未来需要显式建模协议阶段结构，并约束生成的 correspondence lemmas 在阶段点与时序关系上与协议语义一致。最后，若面向实际漏洞发现，还需打通“符号攻击轨迹”到“现实语义漏洞”的映射：即如何将 Tamarin 返回的攻击路径解释为具体实现/部署中的可触发条件与安全影响，从而形成可操作的漏洞报告。
Our findings suggest several directions for future work. 
First, CRoST can still produce false positives when proof paths contain common low-information solve items, such as \texttt{solve(!KU(x))}; an unrelated generated lemma may cover the target lemma at such points.
Future work can introduce stricter constraints into CRoST computation---such as semantic-category filtering of solve items, context-consistency checks, or minimum-information thresholds---to reduce erroneous matches on common solve items and obtain more precise coverage estimates.

% First, addressing the limitation noted in Section~\ref{subsec:limitation}, future work can introduce stricter constraints into the \textsc{CRoST} computation---such as semantic-category filtering of solve items, context-consistency checks, or minimum-information thresholds---to reduce erroneous matches on common solve items over shorter proof paths and obtain more precise coverage estimates.

Second, scaling CSPV to large, complex protocols requires addressing stage point confusion: future methods should explicitly represent protocol phase structure and constrain generated correspondence lemmas to align stage-specific events and temporal relations with protocol semantics.

Third, for real-world vulnerability discovery, an additional step is needed to provide counterexample explanation.

Fourth, incorporating richer features of protocols and security properties may further improve LLM-based lemma generation. Our current study is only a preliminary evaluation and does not introduce additional protocol-level or property-level features, leaving this as an important direction for future work.

% 中文注释：此外，CRoST 本身仍有优化空间。当前的路径匹配可能在存在“过于简单/通用”的 solve 结点时产生少量误报：例如 \texttt{solve(!KU(\~k))} 这类容易被无关的生成 lemma 在某个分支中偶然覆盖，从而提高覆盖计数但不代表真实能力提升。未来可在 CRoST 计算中加入更严格的约束（例如对 solve-item 的语义类别过滤、上下文一致性检查或最小信息量阈值），以获得更精确的覆盖估计。

% Besides, the current \textsc{CRoST} matching can admit occasional false positives when proof skeletons contain overly simple or generic solve nodes.
% For instance, a target solve obligation such as \texttt{solve(!KU(\~k))} may be spuriously matched by an unrelated generated lemma that happens to pass through a similar trivial obligation on some branch, inflating coverage without reflecting meaningful progress.
% A possible solution is to introduce additional constraints in \textsc{CRoST} computation to obtain more precise coverage estimates.

\section{Related Work}
\label{sec:related_work}

% 中文注释：现有符号化协议验证工具确实提供了一定程度的“一键式”自动化，但这类自动化主要针对粗粒度、预定义的属性模板，难以覆盖场景化的高层安全保证。
Automatic security proofs for cryptographic protocols are not a new research problem \cite{10.1007/11513988_27}. Symbolic protocol verifiers provide some basic automation, but it is largely limited to coarse, pre-defined property templates rather than scenario-specific, high-level guarantees.

% 中文注释：例如，\emph{Scyther} 可以自动实例化通用的认证与机密性声明（尤其是对本地生成值与协议变量的 secrecy/authentication 检查），从而在不完全手写每条 claim 的情况下完成基础验证。但这种自动化在很大程度上局限于通用的、粗粒度的安全目标，并不涉及特定场景的复杂协议。
For example, \emph{Scyther}~\cite{10.1007/978-3-540-70545-1_38} can automatically instantiate generic secrecy and authentication claims for protocol variables and locally generated values, enabling basic checks without fully hand-writing every claim. But this automation is largely confined to generic, coarse-grained security goals and does not address scenario-specific, complex protocols.

% 中文注释：类似地，\emph{Tamarin} 提供了面向证明效率的自动辅助机制，例如 auto-sources lemmas，可自动生成用于证明“受保护子项（如哈希/加密内部秘密）机密性”的辅助引理，并缓解部分解构带来的搜索爆炸。然而，这类自动化主要解决“证明搜索效率/可扩展性”问题，并不能替代属性工程中的核心难点：确定“应该验证哪些高层安全属性”，以及将其精确编码为可被工具检验、且在给定模型与威胁假设下可证明（或可反驳）的 lemmas；完成一套完整分析流程在实践中仍高度依赖专家经验与反复迭代。
Similarly, \emph{Tamarin}~\cite{tamarin-manual} supports automation aimed at proof scalability, such as auto-sources lemmas that are automatically generated to facilitate proving secrecy of protected subterms (e.g., secrets occurring under hashes/encryptions) and to mitigate proof-search blowups due to partial deconstructions~\cite{10.1007/978-3-030-59013-0_1, 10.1145/2810103.2813662}.
However, such automation primarily targets proof-search scalability and does not resolve the central property-engineering challenge: deciding which high-level guarantees should be verified in a given scenario and encoding them precisely as tool-checkable lemmas that are provable or refutable under the chosen protocol model and verification setting.

% 中文注释：近期开始出现将 LLM 与符号协议验证结合的工作，重点是降低“建模门槛”：从自然语言文档/协议描述自动生成可被验证工具接受的符号模型。其中，ICSE'25 的 \emph{LLM-aided Automatic Modeling for Security Protocol Verification} 提出从协议文档逐步生成符号模型（以 SAPIC+/多重集覆写规则为核心表示），并通过阶段化设计与形式化变换约束来提升模型正确性；其主要目标是自动生成协议的 rewriting rules（建模），而非自动生成安全属性/lemmas（属性工程）。
Recent work has begun to combine LLMs with symbolic protocol verification to reduce the \emph{modeling} barrier, i.e., generating tool-checkable symbolic models directly from natural-language protocol descriptions.
\cite{10.1109/ICSE55347.2025.00197} studies how to synthesize symbolic protocol models (centered on SAPIC+ and multiset rewriting rules) from natural-language documents via a staged pipeline with formally constrained transformations.
This line of work primarily targets the generation of protocol rewriting rules (model construction), rather than the generation of security properties/lemmas (property engineering). 

CryptoFormalEval \cite{curaba2024cryptoformaleval} provides a simple evaluation of LLMs in symbolic reasoning. However, they do not provide a method for leveraging LLM capabilities beyond syntax, so their results remain at the level of syntax correctness.

% 中文注释：与之互补地，我们关注在“模型已给定”的情况下，LLM 能否生成可被工具检验的 lemmas，并系统刻画其在不同属性类型与复杂协议上的能力与失败模式。
Complementary to those directions, our work focuses on the setting where the protocol model is already available and examines how well LLMs can generate tool-checkable lemmas.

\section{Conclusion}
\label{sec:conclusion}
% 中文注释：本工作对 LLM 在 CSPV（以 lemma 生成 + 工具检验为核心）的能力进行了系统性实证刻画，并得到若干稳定结论：多数模型需要 few-shot 才能稳定生成可编译的 Tamarin lemmas；不同模型与不同属性类型的 CRoST 差异显著，且在复杂多阶段协议上容易出现阶段点混淆等失败模式；在固定提示与流程下，简单增加生成数量会导致覆盖率收益递减并出现饱和；同时，LLM 生成 lemmas 的验证开销显著高于基线 target lemmas，说明可用性受验证成本制约。
Our work shows that few-shot prompting is often necessary to obtain syntactically valid Tamarin lemmas; \textsc{CRoST} varies substantially across LLMs and property families, and complex multi-phase protocols exhibit recurrent failure modes such as stage point confusion.
Moreover, under fixed prompting and pipeline settings, increasing the number of generated lemmas yields diminishing returns and coverage saturation.
Finally, verifying LLM-generated lemmas incurs substantially higher runtime than verifying baseline target lemmas, highlighting verification cost as a practical constraint.

\section*{Data Availability}

The data supporting the results of this paper are publicly available at Zenodo\footnote{https://doi.org/10.5281/zenodo.20996406}.

% \section*{Acknowledgments}
% ChatGPT was utilized to generate sections of this Work, including text, tables, graphs, code, data, citations, etc.

\bibliographystyle{IEEEtran}
\bibliography{sample-base}

\end{document}